\documentclass[12pt]{article}

\usepackage{xr}

\usepackage[inline,shortlabels]{enumitem}
\usepackage{float}
\usepackage{mathtools}
\usepackage{subdepth}
\usepackage{comment}
\usepackage{xspace}
\usepackage{graphicx}
\usepackage[round]{natbib}
\usepackage[margin=1in]{geometry}
\usepackage{tikz}
\usetikzlibrary{arrows.meta}
\usetikzlibrary{decorations.markings}
\usepackage[linesnumbered]{algorithm2e}
\RestyleAlgo{ruled}
\usetikzlibrary{plotmarks}
\usetikzlibrary{shapes,decorations,arrows,calc,arrows.meta,fit,positioning}
\usetikzlibrary{shapes.geometric}
\tikzset{
font={\fontsize{7pt}{12}\selectfont},
every node/.style={scale=1.3},
square/.style={regular polygon,regular polygon sides=4},
    -Latex,auto,node distance =1 cm and 1 cm,semithick,
    state/.style ={ellipse, draw, minimum width = 0.7 cm},
    block/.style = {square, draw, inner sep=0cm,minimum size=8mm},
    free/.style = {circle, draw, inner sep=0cm,minimum size=6mm},
    bidirected/.style={Latex-Latex,dashed},
    el/.style = {inner sep=2pt, align=left, sloped}
}\usepackage[english]{babel}
\usepackage{longtable}
\usepackage{nccmath}
\usepackage{color}
\usepackage{mathtools}
\usepackage{amssymb,amsmath,amsthm}
\usepackage{multirow}
\usepackage[titletoc,title]{appendix}
\usepackage{authblk}
\usepackage{bbm}
\usepackage{setspace}
\usepackage{dsfont}
\usepackage[OT1]{fontenc}
\usepackage{subcaption}
\usepackage{refcount}
\usepackage{booktabs}
\usepackage{setspace}
\graphicspath{ {../simulations/01_init_manuscript/graphs/} }
\usepackage{xr}
\usepackage[colorlinks,citecolor=blue,urlcolor=blue]{hyperref}

\newtheorem{theorem}{Theorem}
\AtEndDocument{\refstepcounter{theorem}\label{finalthm}}
\AtEndDocument{\refstepcounter{equation}\label{finaleq}}

{
  \theoremstyle{definition}
  \newtheorem{assumption}{}
}
{
  \theoremstyle{definition}
  
}
{
  \theoremstyle{definition}
  \newtheorem{example}{Example}[section]
}
{
  \theoremstyle{definition}
  \newtheorem*{example*}{Example}
}
{
  \theoremstyle{definition}
  \newtheorem{condition}{Condition}
}

\newtheorem{lemma}{Lemma}
\AtEndDocument{\refstepcounter{lemma}\label{finallemma}}

\newtheorem{definition}{Definition}

\renewcommand{\P}{\mathsf{P}}

\newcommand{\F}{\mathsf{F}}

\newcommand{\Rem}{\mathsf{R}}

\newcommand{\m}{\mathsf{m}}
\newcommand{\g}{\mathsf{g}}

\newcommand{\G}{\mathsf{G}}

\newcommand{\indep}{\mbox{$\perp\!\!\!\perp$}}

\newcommand{\D}{\mathsf{D}}

\newcommand{\one}{\mathds{1}}
 
\newcommand{\E}{\mathsf{E}}
\renewcommand{\P}{\mathsf{P}}
\newcommand{\Ec}{\mathbb{E}}

\newcommand{\1}{\mathbbm{1}}

\newcommand{\supp}{\mathop{\mathrm{supp}}}

\DeclareMathOperator*{\argmin}{\arg\!\min}
\DeclarePairedDelimiterX{\norm}[1]{\lVert}{\rVert}{#1}

\usepackage{accents}

\renewenvironment{proof}{{\it Proof }}{\qed \\}

\pgfdeclarelayer{background}
\pgfsetlayers{background,main}
\usetikzlibrary{arrows,positioning}
\tikzset{
>=stealth',
punkt/.style={
rectangle,
rounded corners,
draw=black, very thick,
text width=6.5em,
minimum height=2em,
text centered},
pil/.style={
->,
thick,
shorten <=2pt,
shorten >=2pt,}
}
\newcommand{\Vertex}[3]
{\node[minimum width=0.6cm,inner sep=0.05cm] (#2) at (#1) {#3};
}
\newcommand{\Vertexr}[3]
{\node[rectangle, draw, minimum width=0.6cm,inner sep=0.05cm] (#2) at (#1) {#2};
}
\newcommand{\ArrowR}[3]%
{ \begin{pgfonlayer}{background}
\draw[->,#3] (#1) to[bend right=30] (#2);
\end{pgfonlayer}
}
\newcommand{\ArrowLW}[3]%
{ \begin{pgfonlayer}{background}
\draw[->,#3] (#1) to[bend left=30] (#2);
\end{pgfonlayer}
}

\newcommand{\ArrowL}[3]%
{ \begin{pgfonlayer}{background}
    \draw[->,#3] (#1) to[bend left=45] (#2);
  \end{pgfonlayer}
}

\newcommand{\EdgeL}[3]%
{ \begin{pgfonlayer}{background}
\draw[dashed,#3] (#1) to[bend right=-45] (#2);
\end{pgfonlayer}
}

\newcommand{\Arrow}[3]%
{ \begin{pgfonlayer}{background}
\draw[->,#3] (#1) -- +(#2);
\end{pgfonlayer}
}

\allowdisplaybreaks
\pgfarrowsdeclare{arcs}{arcs}{...}
{
  \pgfsetdash{}{0pt} 
  \pgfsetroundjoin   
  \pgfsetroundcap    
  \pgfpathmoveto{\pgfpoint{-5pt}{5pt}}
  \pgfpatharc{180}{270}{5pt}
  \pgfpatharc{90}{180}{5pt}
  \pgfusepathqstroke
}
\newcommand{\ArrowB}[3]%
{ \begin{pgfonlayer}{background}
    \draw[|-arcs,line width=0.4mm,shorten <= 0.3cm,shorten >= 0.3cm,#3] (#1) -- +(#2);
  \end{pgfonlayer}
}

\newcommand{\GATT}{GATT\xspace}
\newcommand{\GATTs}{GATTs\xspace}

\date{\today}

\mathtoolsset{showonlyrefs}
\usepackage{tikzscale}

\usepackage{natbib}
\usepackage{xr}

\newcommand{\figurestables}{1}

\ifx\arxiv\undefined 
    \date{}
\fi
\ifx\anonymized\undefined
    \author[1,*]{Herbert Susmann}
    \author[2]{Nicholas T. Williams}
    \author[2]{Kara E. Rudolph}
    \author[1]{Iv\'an D\'iaz}

    \date{}
    
    \affil[1]{\small Division of Biostatistics, Department of Population
      Health, New York University Grossman School of Medicine, New York, NY, USA}
    \affil[2]{\small Department of Epidemiology, Mailman School of Public Health, Columbia University, New York, NY, USA.}
    \affil[*]{Corresponding author: susmah01@nyu.edu}
\else
    \author{}
\fi

\title{Longitudinal Generalizations of the Average Treatment Effect on the Treated for Multi-valued and Continuous Treatments}

\begin{document}
\maketitle
\begin{abstract}
The Average Treatment Effect on the Treated (ATT) is a common causal parameter defined as the average effect of a binary treatment among the subset of the population receiving treatment. We propose a novel family of parameters, Generalized ATTs (\GATTs), that generalize the concept of the ATT to longitudinal data structures, multi-valued or continuous treatments, and conditioning on arbitrary treatment subsets. We provide a formal causal identification result that expresses the \GATT in terms of sequential regressions, and derive the efficient influence function of the parameter, which defines its semi-parametric efficiency bound. Efficient semi-parametric inference of the \GATT requires estimating the ratios of functions of conditional probabilities (or densities); we propose directly estimating these ratios via empirical loss minimization, drawing on the theory of Riesz representers. Simulations suggest that estimation of the density ratios using Riesz representation have better stability in finite samples. Lastly, we illustrate the use of our methods to evaluate the effect of chronic pain management strategies on the development of opioid use disorder among Medicare patients with chronic pain. 
\end{abstract}


\ifx\arxiv\undefined 
    \clearpage
\fi

\section{Introduction}

Many causal estimands of scientific interest are defined in terms of contrasts of marginal means of counterfactual outcomes under different treatment assignments, averaged over a population. However, in many cases it is also of interest to study the effect of a treatment only among those who received it \citep{heckman2001fourparameters}. For instance, in single time point settings with a binary treatment the well-known Average Treatment Effect (ATE) is the expected difference in counterfactual outcomes under treatment vs.~under control averaged over the entire population, while the Average Treatment Effect on the Treated (ATT) averages only over the subpopulation that received the treatment. In more detail, with $A \in \{0, 1\}$ a binary treatment variable and $Y(0)$ and $Y(1)$ representing counterfactual outcomes under each treatment, the ATT is defined as
 \(   \E[Y(1) - Y(0) \mid A = 1]\). 
As causal effects analogous to the ATE can be defined in more complex data structures and for more complex interventions, so can the ATT be generalized. For example, in a simple longitudinal setting with two time points where $A_1$, $A_2$ are binary treatment indicators at times $1$ and $2$ and $Y(a_1, a_2)$ is the counterfactual outcome under treatment assignments $a_1$ and $a_2$, one possible generalization of the ATT is the parameter given by
 \(   \E[Y(1, 1) - Y(0, 0) | A_1 = 1]\). 
That is, the parameter is the difference in counterfactual outcomes under treatment at both time points vs. no treatment at both time points among the population who received the treatment at the first time point.
Further generalizations are also possible: for example, we may be interested in alternative treatment trajectories, or in conditioning on alternative treatment subsets. Going farther, we may further generalize to longitudinal structures with more time points, and to treatments that are multi-valued or continuous.

In order to cover a large range of possible interventions of interest defined for arbitrary longitudinal structures, we work within the framework of Longitudinal Modified Treatment Policies (LMTPs), which cover a broad class of causal effects for continuous, binary, and time to event outcomes \citep{robins2004effects, Diaz12, Haneuse2013,young2014identification,diaz2023lmtp, hoffman2023introducinglmtps,diaz2024causal}. LMTPs are defined as the population expected counterfactual outcome under an intervention 
that changes the natural value of the treatment (that is, the treatment an individual would have received under no intervention; \citealt{young2014identification}) according to a user-given function, which is called a Modified Treatment Policy (MTP). For example, an MTP for a continuous treatment could be defined as a function that shifts the natural value of the treatment multiplicatively or additively by a fixed amount. The novelty of MTPs compared to dynamic treatment rules, which define treatment as a function of time-varying covariates and prior treatment status, is that MTPs are allowed to depend on the natural value of the treatment, therefore enabling analysts to define diverse causal effects for complex (continuous, multivariate, time-to-event, etc.) treatments \citep{hoffman2023introducinglmtps}.

We draw on the LMTP framework to propose a novel class of causal parameters, referred to as Generalized ATTs (\GATTs), defined as the expected counterfactual outcome under an MTP intervention conditional on the natural value of treatment falling in an arbitrary conditioning set. The ATT for a time point treatment is a notable member of the class of \GATT parameters. Crucially, the \GATT is conditional on the longitudinal natural value of treatment: the counterfactual treatments that an individual would have received at each time point had they followed the treatment policy up to that time. Conditioning on the longitudinal natural value of treatment, rather than the \textit{observed} treatment trajectory, is crucial in establishing that the family of \GATT parameters are identifiable. Drawing on this distinction, we continue by showing that the \GATT is causally identifiable with an identification result that extends the g-formula in terms of sequential regressions. These results generalize the identification formulas for longitudinal parameters \citep{richardson2013single,young2014identification,diaz2023lmtp} and for the ATT \citep{luedtke2017sequential,Bang05,hubbard2011,heckman1995iv}.

We introduce several estimators of the \GATT, including an asymptotically normal and efficient estimator based on Targeted Minimum Loss-Based Estimation (TMLE). To begin, we derive simple substitution and probability (density) ratio estimators, the latter being similar in spirit to the well-known inverse probability weighted estimators (IPW) for parameters such as the ATE.
However, the performance of these simple estimators depends on consistent estimation of the corresponding nuisance parameters (the outcome and treatment assignment mechanisms, respectively). If data-adaptive methods are used for nuisance estimation, which is desirable in order to avoid making parametric assumptions, then the substitution and reweighting estimators have non-negligible asymptotic bias, and their sampling distributions are generally unknown. Therefore, we draw on tools from semi-parametric efficiency theory to construct novel estimators that are unbiased and asymptotically efficient even when using data-adaptive nuisance estimation algorithms. In particular, we derive the so-called \textit{von-Mises} expansion \citep{mises1947asymptotic} of the \GATT parameter, decomposing estimators thereof into first- and second-order bias terms. The first-order bias term is equal to the average of the parameter's canonical gradient, also known as its \textit{efficient influence function}, which characterizes the semi-parametric efficiency bound of regular asymptotically linear estimators of the \GATT \citep{Bickel97}. A major contribution of this work is to derive the form of the EIF for the \GATT (Theorem~\ref{thm:eif}). Consequently, we are able to construct an estimator that achieves the semi-parametric efficiency bound using TMLE \citep{vanderLaanRose11,hubbard2011}. The TMLE approach works by fluctuating initial nuisance parameter estimates, which may be derived from data-adaptive methods, in such a way that the resulting fluctuated estimates solve the the EIF estimating equation. 

The EIF of the \GATT, which subsequently determines the form of the TMLE estimator, has an additive structure in which terms for each time point are weighted by a probability ratio (or density ratio, for continuous treatments). We show that these weighting terms can be interpreted as \textit{Riesz representers}, connecting our methods to the Double Machine Learning literature \citep{chernozhukov2018dml}. The Riesz representer for each time point is the cumulative multiplication of all previous density ratios. One approach to estimating the required probability (density) ratios for these cumulative weights is to first estimate the required conditional probabilities (densities) and then multiply them together. However, this can lead to significant instability in the resulting estimators, especially in cases with high-dimensional or continuous treatments and in long longitudinal data structures. 
\cite{diaz2023lmtp} applied an approach based on estimating the density ratios and subsequently multiplying them cumulatively, but this approach does not neatly generalize to \GATTs and can also suffer from instability in the resulting estimators. Thus, another major contribution is to propose a strategy that estimates the weights directly through empirical loss minimization \citep{chernuzhukov2022riesznet}, drawing on their interpretation as Riesz representers. This approach avoids the instabilities inherent when weights are formed by inverting and cumulating probabilities (densities).  The good performance of this approach is of wider general interest beyond estimating \GATTs, and applies to any estimator that requires inverse probability (density) weights or ratios, and especially to longitudinal settings involving cumulative ratios.

The overall goal of the present work is to propose an end-to-end methodology for the definition, identification, and estimation of the effects of modified treatment policies within strata of the population defined by values of the treatment.
Our main contributions are threefold. Our first contribution is to define a novel \GATT parameter in terms of modified treatment policies and establish causal identification results. Our second contribution is to derive key theoretical quantities such as the EIF of the \GATT parameter, which determines its semi-parametric efficiency bound. We then propose a semi-parametric efficient estimator based on TMLE and propose an estimation strategy involving empirical loss minimization based on Riesz representers, which is our third major contribution. The rest of the manuscript unfolds as follows: in the remainder of the introduction we discuss related prior work. In Section~\ref{section:causal-effects} we rigorously define the \GATT within a causal structural equation model and provide the causal identification result. In Section~\ref{section:efficiency} we analyze the \GATT statistical parameter in the framework of semi-parametric efficiency theory and derive its canonical gradient, its efficiency bound, and properties that we expect to hold for estimators based on the canonical gradient such as multiple robustness. In Section~\ref{section:estimation} we present an estimation strategy based on TMLE and direct estimation of the cumulated probability weights. The performance of the methods are investigated via simulation studies in Section~\ref{section:simulation}, and an applied example based on real world data is provided in Section~\ref{section:application}. We conclude with a discussion in Section~\ref{section:discussion}.

\subsection{Prior Work}
The ATT has been extensively studied in the case of a single time point and with binary treatments \citep{heckman1995iv, hahn1998role, shpitser2009, leacy2014att, wang2017, matsouaka2023att}. Doubly-robust estimators of the ATT have been proposed based on TMLE \citep{hubbard2011} and augmented inverse probability weighting \citep{moodie2018att}.  The ATT has also been extended to the case of categorical treatments \citep{vanderweele2013treatment}. In the longitudinal setting, parameters analogous to the ATT have been studied in the context of instrumental variable designs \citep{tchetgentchetgen2013att,liu2015msms} and in a more general setting using Marginal Structural Models \citep{schomaker2023doubly}. 

Longitudinal causal parameters based on modified treatment policies grew out of earlier work on defining and estimating causal effects in longitudinal settings and with complex interventions and non-binary treatments. \cite{Diaz12} proposed doubly robust estimators for causal effects defined via shift interventions for single time points. This work was further developed by \citep{Haneuse2013}, who introduced the term \textit{modified treatment policy}. From another angle, there was research on dynamic treatment regimes in longitudinal settings  \citep{robins2004effects, richardson2013single, young2014identification}.
These lines of research can be synthesized into a general framework for defining and estimating longitudinal causal effects, referred to as Longitudinal Modified Treatment Policies (LMTPs; \citealt{diaz2023lmtp}).

The definition of our causal parameters relies on the structural causal models of \cite{pearl2009}, but analogous definitions could have been achieved under a potential outcomes framework \citep{richardson2013single}. The  analysis of the properties of the statistical identification formula, as well as the development of estimators, builds on a long line of research in general semi-parametric efficiency theory \citep{bickel1982adaptive, vanderVaart&Wellner96}, and more recent results for longitudinal settings \citep{luedtke2017sequential, rotnitzky2017multiply}. The TMLE framework for constructing efficient estimators of parameters within semi-parametric models is described thoroughly in \cite{vdl2006targeted,vanderLaanRose11}. Our approach for estimating cumulated probability (density) ratios is inspired by recent research related to the use of Riesz representers in statistical estimation \citep{chernozhukov2021automatic, chernuzhukov2022riesznet,chernuzhukov2022debiased,chernozhukov2023automatic}.

\section{Causal Effects}
\label{section:causal-effects}
Let $Z = (L_1, A_1, L_2, A_2, \dots, L_\tau, A_\tau, Y)$ be a random variable where, for $t \in \{1, \dots, \tau\}$, $L_t$ is a vector of time-varying covariates, $A_t$ a vector of categorical, multivariate, or continuous treatment variables, and $Y$ a binary or continuous outcome. Suppose we have a sample $Z_1, \dots, Z_n$ of of i.i.d. draws $Z \sim \P_0$, where $\P_0$ falls in a non-parametric statistical model $\mathcal{M}$. Denote the history and future of a random variable as $\bar{X}_t = (X_1, \dots, X_t)$ and $\underline{X}_t = (X_t, \dots, X_\tau)$, respectively. For succinctness, we will write $\bar{X}$ and $\underline{X}$ to denote the complete history and future of a random variable ($\bar{X}_\tau$ and $\underline{X}_\tau$, respectively). Let $H_t = (\bar{A}_{t-1}, \bar{L}_t)$ be the history of all variables until right before $A_t$. Let $g_t(a_t, h_t)$ be the probability density or probability mass function of $A_t$ conditional on $H_t = h_t$, evaluated at $a_t$. We use $\supp\{\cdot\mid \mathcal B\}$ to denote the support of a random variable conditional on the set $\mathcal B$.

The causal model, which describes the data-generating process of the observed data and defines counterfactual outcomes, is formalized via a non-parametric structural equation model \citep{pearl2009}. Within this framework, we assume the observed data are generated according to the following deterministic functions, for $\{ 1, \dots, \tau \}$:
\begin{align}
    L_t &= f_{L_t}(A_{t-1}, H_{t-1}, U_{L,t}), \\
    A_t &= f_{A_t}(H_t, U_{A,t}), \\
    Y   &= f_Y(A_{\tau}, H_{\tau}, U_Y),
\end{align}
where $U = \left(U_{L, t}, U_{A, t}, U_Y : t \in \{ 1, \dots, \tau \} \right)$ is a set of exogenous variables. Note that the model implies a particular time-ordering of the variables: each $L_t$ happens before the corresponding $A_t$, and $Y$ occurs last.
\allowdisplaybreaks 
Interventions are defined by replacing $A_t$ with a new random variable $A_t^d$; we will define what this means in more detail shortly. An intervention $\bar{A}_{t-1}^d$ on all treatments from time 1 to time $t-1$ induces counterfactual variables $L_t(\bar{A}^d_{t-1}) = f_{L_t}(A_{t-1}^d, H_{t-1}(\bar{A}_{t-2}^d, U_{L,t})$ and $A_t(\bar{A}_{t-1}^d) = f_{A_t}(H_t(\bar{A}_{t-1}^d), U_{A, t})$, with counterfactual history defined as $H_t(\bar{A}_{t-1}^d) = (\bar{A}_{t-1}^d, \bar{L}_t(\bar{A}_{t-1}^d))$. The counterfactual variable $A_t(\bar{A}_{t-1}^d)=f_{A_t}(H_t(\bar{A}_{t-1}^d), U_{A,t})$ is referred to as the \textit{natural value of treatment} \citep{young2014identification}, and is interpreted as the value treatment would have taken had the intervention been implemented but discontinued right before time $t$. Intervening on all treatment variables up to time $t = \tau$ induces the counterfactual outcome $Y(\bar{A}^d) = f_Y(A_\tau^d, H_\tau(\bar{A}_{\tau - 1}^d), U_Y)$. We focus on LMTP interventions, which are a particular type of intervention in which $A_t^d$ is defined as a function of the natural value of treatment at time $t$ and the complete counterfactual history.\allowdisplaybreaks 

\begin{definition}{\cite[Definition 1,][]{diaz2023lmtp}}
The intervention $A_t^d$ is an LMTP if it has a representation $A_t^d = d(A_t(\bar{A}_{t-1}^d), H_t(\bar{A}_{t-1}^d))$ for an arbitrary function $d$.  
\end{definition}

Shift interventions, in which continuous treatments are shifted by a fixed $\delta$, are a popular example of an LMTP originally proposed in \cite{Diaz12} and discussed further in \cite{diaz2018stochastic}, \cite{Haneuse2013}, and \cite{hoffman2023introducinglmtps}. Shift interventions are carefully designed so that the shifted treatment cannot fall outside the range of empirically observed treatment values; this avoids the need for extrapolating counterfactual outcomes to treatments that have never been observed. A formal example of a positive shift intervention is given below.
\begin{example*}[Shift LMTP]
    Suppose there exists some $u_t$ such that $\P(A_t > u_t | H_t = h_t) = 1$ for all $t \in \{ 1, \dots, \tau \}$. For some fixed $\delta > 0$, define the intervention as
    \begin{align}
        d(a_t, h_t) = \begin{cases}
            a_t + \delta, &\text{ if } a_t \leq u_t(h_t) - \delta, \\
            a_t, &\text{ if } a_t > u_t(h_t) - \delta.
        \end{cases}
    \end{align}
\end{example*}
The counterfactual outcome induced by the above LMTP can be interpreted as the counterfactual outcome had every individual experienced a treatment shifted upwards by $\delta$, so long as the shifted treatment falls within the range of observed treatments; otherwise, the treatment remains at its natural value. 

Existing literature on LMTPs focuses on defining, identifying, and estimating population-level parameters such as the population mean counterfactual outcome under an LMTP, expressed as $\E[Y(\bar{A}^d)]$. However, no methodology currently exists for cases where parameter of interest is the \textit{conditional} mean counterfactual outcome of an LMTP for a subpopulation defined by treatment status. In the sequel, we present identification and estimation theory for parameters that condition on a subset of the treatment space and provide motivating examples where such parameters are of scientific relevance. 

\subsection{Definition of Causal Effects}
In this section we formalize the causal effect of an MTP intervention conditional on a set of treatment statuses. Our definition depends on the crucial concept of the longitudinal natural values of treatment, which we set apart below for emphasis:

\begin{definition}[]
    The vector $\bar A(d) = (A_1, A_2(d_1), ..., A_\tau(d_{\tau-1}))$ is called the longitudinal natural value of treatment.
\end{definition}
In other words, the vector $A(d)$ contains the counterfactual value that the treatment would have taken at each time point if the MTP had been followed up to but not including that time point. For the first time point, there are no previous treatments so there is no counterfactual natural value of treatment.

We are now ready to define the causal effect of interest. Let $\bar{\mathcal{B}} \subset \bar{\mathcal{A}}$ be a subset of the space of all possible longitudinal treatment assignments. 
Define the counterfactual outcome average conditional on the longitudinal natural value of treatment falling in $\bar{\mathcal{B}}$ as
\begin{align}
    \label{eq:gatt}
    \theta^* = \Ec\left\{ Y(\bar{A}^d) \mid \bar{A}(d) \in \bar{\mathcal{B}} \right\}.
\end{align}
We refer to $\theta^*$ as a \GATT, as it generalizes the single-time point ATT parameter to longitudinal data structures, multi-valued or continuous treatments, and interventions defined in terms of an MTP. 
The parameter reduces to an unconditional LMTP parameter as studied in prior literature if $\bar{\mathcal{B}} = \bar{\mathcal{A}}$ (that is, if the parameter conditions on the space of all possible treatment assignments).

Note that the parameter conditions on the \textit{counterfactual} treatment trajectory $\bar{A}(d)$, rather than on the \textit{observed} treatment trajectory $\bar{A}$; this is crucial to ensure that the parameter is identifiable. For intuition as to why this is the case, note that treatments $A_t$ for $t > 1$ are \textit{mediators} of the effect of $A_1$ on $Y$. If we were to instead define our parameter by conditioning on $\bar{A} \in \bar{\mathcal{B}}$ (that is, we condition on the \textit{observed} treatment trajectory falling in $\bar{\mathcal{B}}$) then the parameter would condition on the values of a mediator, which in the general case is not an identifiable causal parameter. 

When comparing two different policies $d^\star$ and $d'$ using contrasts between $\Ec[ Y(\bar{A}^{d^\star}) \mid \bar{A}(d^\star) \in \bar{\mathcal{B}}^\star ]$ and $\Ec[ Y(\bar{A}^{d'}) \mid \bar{A}(d') \in \bar{\mathcal{B}}' ]$ it is important to ensure that the sub-populations they are based on are comparable, i.e., that the regimes $d^\star$ and $d'$ and the sets $\bar{\mathcal{B}}'$ and $\bar{\mathcal{B}}^\star$ are such that $\bar{A}(d^\star) \in \bar{\mathcal{B}}^\star$ and $\bar{A}(d') \in \bar{\mathcal{B}}'$ if and only if $\bar{A}(d^\star) \in \bar{\mathcal{B}}^\star$. A simple sufficient condition for this to hold is given below; the principle is that is that as soon as the conditioning set begins to involve counterfactual natural values of treatment, the subpopulation induced by the conditioning set may diverge from the subpopulation induced by a different treatment policy. Condition~\ref{condition:comparable} avoids this scenario by requiring a trivial conditioning set as soon as the intervention begins. 
\begin{condition}[Sufficient condition for comparable subpopulations]
    \label{condition:comparable}
    There exists $1 \leq k \leq \tau$ such that $A_t(d) = A$ for all $t < k$ and $\mathcal{B}_t = \mathcal{A}_t$ for $t \geq k$. 
\end{condition}
Next, we give examples of \GATT parameters that can be of scientific interest in certain applications. 
\begin{example}[Single time point average treatment effect on the treated (ATT)]
    Assume the treatment $A$ is binary; for example, $A$ might indicate whether a patient is treated with a particular medication. It is often of interest to estimate the average treatment effect on the treated, given by $\E[Y(1) - Y(0)\mid A=1]$. The average outcome under treatment for the treated population, $\E[Y(1)\mid A=1]=\E[Y\mid A=1]$, can be easily identified and estimated. The average outcome under control for the treated population, $\E[Y(0)\mid A=1]$, is a true counterfactual quantity and requires extra work. In the form of a \GATT, letting $d(A,W)=0$ and $\mathcal B=\{1\}$, this quantity is given by $\theta^* = \E[Y(A^d)\mid A\in \mathcal B]$. Condition~\ref{condition:comparable} is satisfied with $k = 1$.
\end{example}

\begin{example}[Longitudinal ATT for two time points]
    Let $X = (L_1, A_1, L_2, A_2, Y)$ where $A_1$ and $A_2$ are binary. We have at least two options for well-defined conditional effects. First, let $d_t^\star(a_t, h_t) = 1$ and $d_t'(a_t, h_t) = 0$ for $t = 1, 2$. Then, ensuring that the contrast is a well defined causal effect can be achieved by defining $\mathcal{B}_1^\star = \mathcal{B}_1' = \{1\}$ and $\mathcal{B}_2^\star = \mathcal{B}_2' = \mathcal{A}$, yielding the contrast $\E[Y(1,1) - Y(0,0)\mid A_1=1]$ (satisfying Condition~\ref{condition:comparable} with $k = 1$). Alternatively, if $\mathcal{B}_t^\star = \mathcal{B}_t' = \{1\}$ for $t = 1, 2$, defining comparable contrasts can be achieved by letting $d_1^\star(a_1, h_1) = d_1'(a_1, h_1)=a_1$ and $d_2^\star(a_2, h_2) = 1$,  $d_2'(a_2, h_2) = 0$, yielding the contrast $\E[Y(1,1) - Y(1,0)\mid A_1=1, A_2=1]$ (satisfying Condition~\ref{condition:comparable} with $k = 2$). In the following, we focus on identifying and estimating the parameter $\E[Y(0, 0) \mid A_1 = 1]$ as an illustrative example.
\end{example}

\begin{example}[Longitudinal policy-relevant effects]
    For numerical treatments, it is often of scientific interest to investigate what would have been the outcome in a counterfactual world where treatment was increased by some user-given amount. For example, let $A$ denote the particulate matter PM2.5 that a given individual is exposed to. The Environmental Protection Agency sets the National Ambient Air Quality Standards for Particulate Matter to ``protect millions of Americans from harmful and costly health impacts, such as heart attacks and premature death'' \citep{EPA2023PM}. Current standards set the PM2.5 limit to 9 micrograms per cubic meter. Thus, any policy-relevant causal effect would have to take into account this standard. For instance, one may be interested in the effect of reducing PM2.5 on health outcomes by 10\% for geographical areas which are non-compliant with the standard at time $1$, i.e., one may define $\theta^*=\E[Y(d)\mid A_1 \in \mathcal{B}_1]$, where $\mathcal{B}_1 = \{ A_1 : A_1 > b \}$, $\mathcal{B}_t = \mathcal{A}$ for $t > 1$, $d(a_t,h_t) = \delta\times a_t$ with $b=9$, $\delta=0.9$, and $Y$ is a health outcome of interest, for example myocardial infarction. Note that Condition~\ref{condition:comparable} is satisfied with $k = 1$. 
\end{example}

\subsection{Causal Identification}
Causal identification of the \GATT parameter $\theta^*$ is achieved under the following two assumptions:
\begin{assumption}[Positivity]\label{assumption:positivity}
    For all $t \in \{ 1, \dots, \tau \}$, if $(a_t, h_t) \in \supp\{ A_t, H_t \mid A_t(d_{t-1}) \in \mathcal{B}_t \}$ then $(d(a_t, h_t), h_t) \in \supp\{A_t, H_t \}$.
\end{assumption}
\begin{assumption}[Strong sequential randomization]\label{assumption:strong-sequential-randomization} For all $t \in \{1, \dots, \tau \}$,
$U_{A,t} \indep (\underline{U}_{L,t+1}, \underline{U}_{A, t + 1}) | H_t$. 
\end{assumption}
 Positivity assumptions such as \ref{assumption:positivity} are a standard requirement for many causal parameters \citep{petersen2012diagnosing,vanderLaanRose11}. Assumption~\ref{assumption:positivity} is weaker than that required for population-averaged LMTP outcomes, as here positivity is only required conditional on $\mathcal{B}$. This has important implications in applications, as different choices of the set $\mathcal B$ can lead to assumptions with varying degrees of plausibility; see Example 2 below for the ATT, and discussion in \cite{petersen2012diagnosing}. The strong sequential randomization assumption requires that $A_t$, $L_{t+1}$, and $A_{t+1}$ are unconfounded conditional on observed past data (that is, all common causes of $A_t$, $L_{t+1}$, and $A_{t+1}$ are included in $H_t$). This is a stronger assumption that what is required to identify dynamic treatment rules that do not depend on the natural value of treatment. See \cite{diaz2023lmtp} for an in-depth discussion of these identification assumptions, and \cite{richardson2013single, young2014identification} for equivalent assumptions in alternative causal models. Next, we explain the identification assumptions for each of the running examples.
\setcounter{example}{0}
\begin{example}[continued]
    Consider the parameter contrast $\E[Y(1,1) - Y(0,0)\mid A_1=1]$.
    Assumption~\ref{assumption:positivity} requires that for any covariate value $l$ such that $\P(L = l \mid A = 1) > 0$, then it must hold that $\P(A = 0 \mid L = l) > 0$. This is weaker than the positivity assumption required for identification of the ATE, which states that for any covariate value $l$ such that $\P(L=l)>0$, then it must hold that $0<\P(A = 1 \mid L = l) <1$. Assumption~\ref{assumption:strong-sequential-randomization} requires that $U_{A} \indep U_Y | L$; that is, there are no unmeasured common causes of $A$ and $Y$. 
\end{example}
\begin{example}[continued]
Recall that this example focuses on the contrast $\E[Y(1,1) - Y(0,0)\mid A_1=1]$, where the counterfactual requiring identification is $\E[Y(0, 0) \mid A_1 = 1]$. 
The positivity assumption (Assumption~\ref{assumption:positivity}) requires that if $\P(L_1 = l_1 | A_1 = 1) > 0$ then $\P(A_1 = 0 | L_1 = l_1) > 0$, and for any $a_2, h_2$ with $P(A_2 = a_2, H_2 = h_2) > 0$ then $P(0, h_2) > 0$. Assumption~\ref{assumption:strong-sequential-randomization} requires that $L_1$ includes all common causes of $A_1$, $A_2$, $L_2$, and $Y$ and that $L_2$ includes all common causes of $A_2$ and $Y$. 
\end{example}
\begin{example}[continued]
If we assume PM2.5 is measured in a discrete scale, assumption~\ref{assumption:positivity} requires that for all time points, for any pair ($a_t, h_t$) such that $\P(A_t = a_t \mid H_t = h_t, A_1 > b) > 0$, then it must hold that $\P(A_t = \delta a\mid H_t = h_t) > 0$. That is, if there are non-compliant (at baseline) geographical units with history $h_t$ that have a PM2.5 value of $a_t$, then there must also exist units with the same longitudinal history having a PM2.5 value of $\delta a_t$. Assumption~\ref{assumption:strong-sequential-randomization} requires that at each $t$, the longitudinal history $H_t$ includes all common causes of $A_t$ 
and ($A_s, L_s$), $s > t$. 
\end{example}
 
Under these assumptions, we establish an identification result for the \GATT parameter in terms of sequential regressions. 
\begin{theorem}
Let $m_{\tau + 1} = Y$, $A_{\tau + 1} = 1, \mathcal{B}_{\tau + 1} = \{ 1 \}$. In a slight abuse of notation, let $A_{t}^d = d(A_{t}, H_t)$.
Recursively define for $t = \tau, \dots , 1$ the parameters
\begin{align}
    \label{eq:m_t}
    m_t : (a_t, h_t) \mapsto \Ec\left[ m_{t+1}(A_{t+1}^d, H_{t+1}) | A_t = a_t, H_t = h_t, \underline{A}_{t+1} \in \underline{\mathcal{B}}_{t+1} \right],
\end{align}
and let $\theta_t = \Ec\left[ m_t(A_t^d, H_t) | \underline{A}_t \in \underline{\mathcal{B}}_t \right]$. Under \ref{assumption:positivity} and \ref{assumption:strong-sequential-randomization}, the \GATT parameter is identified as $\theta_1 = \Ec\left[ m_1(A_1^d, L_1) | \bar{A} \in \bar{\mathcal{B}} \right]$. 
\end{theorem}

While we will discuss estimation in depth later, we can gain intuition about the identification result for longitudinal settings by introducing a simple estimation strategy based on the sequential regression formulation of the identification result. In particular, $m_\tau$ can be estimated as a regression of $Y$ on $A_\tau$ and $H_\tau$. This regression can then be used to obtain predictions under the hypothetical modified treatment policy $A^d_\tau = d(A_\tau, H_\tau)$, which can in turn be regressed on $A_{\tau - 1}, H_{\tau - 1}$ among individuals with $A_{\tau} \in \mathcal{B}_{\tau}$ to yield an estimate of $m_{\tau-1}$. This procedure can be iterated to $t = 1$, which is averaged over individuals with $\bar{A} \in \bar{\mathcal{B}}$ to yield an point estimate of $\theta$. 

It is also helpful to understand the identification result in the running examples. In the single time point setting the identification result simplifies greatly, as shown in the running example for the ATT.
\setcounter{example}{0}
\begin{example}[continued]
    Recall the counterfactual causal parameter is given by $\theta^* = \E[Y(0) \mid A = 1]$. The identification result implies this parameter is identified by $\theta = \E[ \E[ Y \mid L, A = 0] \mid A = 1]$.
\end{example}
\begin{example}[continued]
    For the counterfactual parameter given by $\E[Y(0, 0) \mid A_1 = 1]$, the identification result simplifies to 
    \begin{align}
        \theta = \E[\E\{\E(Y\mid A_2=0, H_2)\mid A_1=0, L_1 \}\mid A_1=1]. 
    \end{align}
\end{example}
\begin{example}[continued]
    First, define recursively the parameters 
    \begin{align}
        m_t: (a_t, h_t) \mapsto \E[m_{t+1}(\delta A_{t+1}, H_{t+1}) \mid A_t = a_t, H_t = h_t],
    \end{align}
    and recall that $m_{\tau + 1} = Y$. 
    The identification result then implies that $\theta = \E[m_1(\delta A_1, L_1) \mid A_1 > b]$. 
\end{example}

\section{Semi-parametric Efficiency Theory}
\label{section:efficiency}
In this section we investigate the semi-parametric properties of the \GATT parameter $\theta$. These results draw on a long literature in semi-parametric efficiency theory; see \cite{begun1983information, Bickel97, vanderVaart98}, among many others. \cite{kennedy2016efficiency} provides a useful review of the relevant theory specialized to causal inference applications. Our main result is to derive the \textit{efficient influence function} (EIF) of $\theta$, a key object in the semi-parametric analysis of the parameter. The EIF is foundational in the analysis for two reasons: first, the variance of the EIF defines the efficiency bound for estimating $\theta$ in the non-parametric model $\mathcal{M}$ \citep{Bickel97}. Second, knowledge of the EIF is crucial for developing efficient non-parametric estimators of $\theta$ and deriving important properties such as their asymptotic sampling distribution. Developing estimators using the EIF flows from the following expansion of the parameter, sometimes referred to as the \textit{von-Mises expansion} \citep{robins2009quadratic, mises1947asymptotic}: for any $\P, \F \in \mathcal{M}$, decompose $\theta(\P)$ as
\begin{align}
    \label{eq:von-mises}
    \theta(\P)  = \theta(\F) - \E_{\F}\{\D(Z; \P)\} + \Rem(\P, \F),
\end{align}
where $\D$ is the EIF and $\Rem(\P, \F)$ is a second-order term of products of differences between functionals of $\P$ and $\F$. Plugging in an estimate $\hat \P$ for $\P$ and setting $\F = \P_0$, and assuming that the second order estimation error $\Rem(\hat \P, \P_0)$ is small enough, one then obtains an approximation to the bias of a plug-in estimator $\theta(\hat \P)$, namely
$\theta(\hat \P) - \theta(\P_0) \approx -\E_{\P_0}\{\D(O;\hat \P)\}$. A biased corrected estimator may then be constructed by subtracting an estimate of this bias from the plug-in estimator (i.e., the so-called one-step estimator) \citep{pfanzagl1982contributions, emery2000}, or constructing an estimator $\hat{P}$ such that this bias converges to zero (i.e., targeted minimum loss-based estimation) \citep{vanderLaanRose11}. 

The EIF and von-Mises expansion can also be used to derive important properties of non-parametric estimators. The form of the EIF may imply multiple-robustness properties in which only combinations of the nuisance parameters need to be estimated consistently for the estimator of $\theta$ to be consistent. In addition, careful analysis of the form of the remainder term $\Rem(P, F)$ can reveal additional useful properties of non-parametric estimators of $\theta$. For common parameters of interest such as the Average Treatment Effect, the remainder term has a product structure that implies nuisance parameters can be estimated at slow $n^{-1/4}$ rates, allowing for the use of flexible machine-learning algorithms in estimation.

For continuous treatments, we require the following assumption originally introduced by by \cite{Haneuse2013} and common in the LMTP literature.
\begin{assumption}[Piecewise smooth invertibility for continuous treatments \citep{Haneuse2013}]
    \label{assumption:piecewise-invertibility}
    For all $t \in \{1, \dots, \tau \}$, and for all $h_t$, assume that the support of $A_t$ conditional on $H_t = h_t$ is partitionable into subintervals $\mathcal{I}_{t,j} : j = 1, \dots J_t(h_t)$ such that $d(a_t, h_t)$ is equal to some $d_j(a_t, h_t)$ and $d_j(\cdot, h_t)$ has an inverse function that is differentiable with respect to $a_t$. 
\end{assumption}

Before presenting the form of the EIF and the second-order remainder term, we define for convenience several additional nuisance parameters. First, the conditional probability of falling in the conditioning set is given by 
\begin{align}
    \label{eq:g_t}
    G_t(A_t, H_t) = P(\underline{A}_{t+1} \in \underline{\mathcal{B}}_{t+1} \mid A_t, H_t). 
\end{align}
Next, define the cumulated probability (density) ratio $\alpha_{t}$ for $t \in \{1, \dots, \tau \}$ as
\begin{align}
    \label{eq:alpha_t}
    \alpha_{t}(A_t, H_t)=\prod_{k=1}^t r_k(A_k, H_k),
\end{align}
with the ratio at time $t$ defined as
\begin{align}
    r_t(a_t, h_t) &= \frac{g_{t, \mathcal{B}}^d(a_t, h_t)}{g_t(a_t, h_t)},
\end{align}
where, for continuous $A_t$ we define \begin{align}
    g^d_{t,\mathcal B}(A_t\mid H_t) =& \sum_{j=1}^{J_t(H_t)}\frac{G_t(b_{j,t}(A_t, H_t), H_t)}{G_{t-1}(A_{t-1},H_{t-1})}\times \\ &\one\{b_{j,t}(A_t, H_t)\in \mathcal I_{j,t}(H_t), b_{j,t}(A_t, H_t)\in \mathcal B_t\} g_t(b_{j,t}(A_t,H_t)\mid H_t)|b_{j,t}'(A_t,H_t)|.
\end{align}
and where $b_{j,t}(\cdot, h_t)$ denotes the inverse of $d_t(\cdot, h_t)$ in $\mathcal I_{j,t}(h_t)$, and  $b_{j,t}'(\cdot,h_t)$ denotes its derivative. For for discrete $A_t$ we define
\begin{align}
    g^d_{t,\mathcal B}(A_t\mid H_t) = \sum_{a_t\in \mathcal B_t}\one\{A_t=d_t(a_t, H_t)\}\frac{G_t(a_t, H_t)}{G_{t-1}(A_{t-1},H_{t-1})}  g_t(a_t\mid H_t).
\end{align}

We now move on to give the form of the EIF and the second-order remainder term for the \GATT; a detailed derivation is provided in the appendix. Doing so requires expanding our previous notation to emphasize the dependence of the nuisance parameters on the underlying probability law. To that end, for any $\P \in \mathcal{M}$ let $m_{t,\P}$ be the parameter $m_{t}$ \eqref{eq:m_t}, $\G_{t,\P}$ be the parameter $G_t$ \eqref{eq:g_t}, and $\alpha_{t,\P}$ be the parameter $\alpha_t$ \eqref{eq:alpha_t}, all when evaluated at the distribution $\P$. 

\begin{theorem}[Efficient influence function]
\label{thm:eif}
    Assume that $d$ does not depend on $\P_0$, and that Assumption~\ref{assumption:piecewise-invertibility} holds in the case of continuous treatments. 
    Then parameter $\theta_1$ is pathwise differentiable and its EIF is given by 
    \begin{align}
        \D(Z; P) = \sum_{t=0}^\tau \alpha_{t,\P}(A_t, H_t) \frac{\1\{\underline{A}_{t+1} \in \underline{\mathcal{B}}_{t+1}\}}{G_{t,\P}(A_t, H_t)} \left\{ m_{t+1,\P}(A_{t+1}^d, H_{t+1}) - m_{t,\P}(A_t, H_t) \right\}.
    \end{align}
    The second-order remainder term is given by
    \begin{align}
        R(\P,\F) =& -\sum_{t=1}^\tau\E_\P[\{\alpha_{t,\P}(A_t, H_t)-\alpha_{t,\F}(A_t, H_t)\}\{m_{t,\P,\F}(A_t, H_t) - m_{t,\F}(A_t, H_t)\}]\\
        &-\sum_{t=1}^\tau\E_\P\left[\alpha_{t,\F}(A_t, H_t)\left\{1 - \frac{G_{t,\P}(A_t, H_t) }{G_{t,\F}(A_t, H_t) }\right\}\{m_{t,\P,\F}(A_t, H_t) - m_{t,\F}(A_t, H_t)\}\right].
    \end{align}
\end{theorem}
The proof follows directly from the results presented in Appendix~A4
and is therefore omitted.  The non-parametric efficiency bound for regular estimators of $\theta$ is given by the variance of the EIF: $\E_{\P_0}[\D(Z; \P_0)^2]$. This variance is also the minimax efficiency bound in the sense of Theorem 8.11 of \cite{vanderVaart98}. The EIF has a similar form to that of the EIF for unconditional LMTP parameters. The main difference is in the new definition of the ratios $r_t$ that incorporate the conditional structure of the parameter, and the addition of the indicator and inverse probability of future inclusion in the treatment set.

To illustrate these results, we apply Theorem~\ref{thm:eif} to derive the EIFs of the parameters in the running examples.
\begin{example}[continued]
    The efficient influence function for the ATT parameter $\theta$ is
    \begin{align}
        \D(Z;\P)=\frac{\one\{A=0\}}{\P(A=1)}\frac{\g(1\mid L)}{\g(0\mid L)}\{Y-\m(A,L)\}+\frac{\one\{A=1\}}{\P(A=1)}\left\{\m(0, L) - \theta(\P)\right\},
    \end{align}
    which coincides with the efficient influence function given by \cite{hubbard2011} for the ATT. 
\end{example}
\begin{example}[continued]
The efficient influence function for the longitudinal ATT parameter $\theta = \E[Y(0, 0) \mid A_1 = 1]$ is given by
{\footnotesize
\begin{align}
    \D&(X;\P) = \frac{\one\{A_2=0,A_1=0\}}{\P(A_1=1)}\frac{\P(A_2=1\mid H_2)}{\P(A_2\mid H_2)}\frac{\P(A_1=1\mid H_1)}{\P(A_1\mid H_1)}\times[Y - \E(Y\mid A_2, H_2)]\\
    &+\frac{\one\{A_1=0\}}{\P(A_1=1)}\frac{\P(A_1=1\mid H_1)}{\P(A_1\mid H_1)}\times[ \E(Y\mid A_2, H_2)- \E\{\E(Y\mid A_2, H_2)\mid A_1, L_1\}]\\
    &+\frac{\one\{A_1=1\}}{\P(A_1=1)}[\E\{\E(Y\mid A_2, H_2)\mid A_1, L_1 \} -\theta(\P)],
\end{align}}%
\end{example}
\begin{example}[continued]
    The efficient influence function is given by
    {\footnotesize
    \begin{align}
        \D(Z;\P)=&\frac{1}{P(A_1 \in \mathcal{B}_1)} \Bigg[\left( \frac{\one\{A_1>\delta b\}}{\P(A_1 > b)}\frac{\g_1(\delta^{-1}A_1\mid H_1)}{\g_1(A_1\mid H_1)} \right) \{m_{2}(\delta A_{2}, H_{2}) -\m_1(A_1,H_1)\} \\ 
        &+ \sum_{t=2}^\tau \left( \prod_{s=1}^t \frac{\g_s(\delta^{-1}A_s\mid H_s)}{\g_s(A_s\mid H_s)} \right) \{m_{t+1}(\delta A_{t+1}, H_{t+1}) -\m_s(A_t,H_t)\}\\
        &+\one\{A_1 \in \mathcal{B}_1\}\{\m_1(\delta A_1, H_1) - \theta(\P)\}\Bigg].
    \end{align}}%
\end{example}

\section{Riesz representers}
\label{section:riesz}
In this section we introduce a novel interpretation of the cumulated ratios $\alpha_t$ as \textit{Riesz representers}, which leads to novel estimation strategies. 
For context, by the Riesz Representation Theorem we know that any linear bounded functional can be represented as an inner product with respect to a function called a \textit{Riesz representer}. The existence of such a function has been used to construct estimators of statistical functionals of interest in causal inference \citep{chernozhukov2018dml}. \citet{chernozhukov2021automatic} estimated Riesz representers via empirical loss minimization using a novel loss function tailored for the problem. We draw on this theory to aid in estimation of the cumulated ratios $\alpha_t$, which have a complex form in terms of the conditional probabilities (densities) $g_t$ and $g_{t,\mathcal{B}}$. By reinterpreting $\alpha_t$ as being the Riesz representer of a carefully chosen functional, we are able to make use of empirical loss based estimation strategies to aid in directly estimating the ratios $\alpha_t$.  

In order to interpret the cumulative ratios $\alpha_t$ as Riesz representers, we first need to carefully define a set of linear functionals with the property that $\alpha_t$ is their Riesz representer. 
To do so, first define as $b_t(A_t, H_t; m_t) = m_t(A_t^d, H_t)$; that is, $b_t$ represents applying the function $m_t$ to $A_t^d$, the shifted treatment value. Then, for $t=0,\ldots, \tau$, let
\begin{align}
    \Psi_t(m_t)=\E_\P\bigg\{\E_\P\big\{\E_\P[b_t(A_t, H_t; m_t)\mid A_{t-1}=A_{t-1}^d, H_{t-1}, \underline{A}_{t} \in \underline{\mathcal{B}}_{t}]\mid A_{t-2} = A_{t-2}^d,  H_{t-2}, \underline{A}_{t-1} \in \underline{\mathcal{B}}_{t-1}\big\}\cdots\bigg\},
\end{align}
where $m_0$ is defined as $\E_\P\{m_1(A_1^d, H_1)\mid \underline A_1\in \underline{\mathcal B}_1\}$ and $\Psi_0(m_1)=m_0$.  
Note that the functionals $\Psi_t$ are linear in $m_t$. Therefore, by the Riesz Representation Theorem, it follows that there exists a function $f_t$ such that
\begin{align}
    \label{eq:riesz-representer}
    \Psi_t(m_t)=\E_\P[f_t(A_t, H_t)m_t(A_t, H_t)]. 
\end{align}
The key property of the functional $\Psi_t$ is that its Riesz representer is precisely the cumulative ratio $\alpha_{t}$, which we establish formally in the following lemma.

\begin{lemma}\label{lemma:riesz}
The Riesz representer for $\Psi_t$, $t = 0, \dots, \tau$ is equal to  $\alpha_{t}$. 
\end{lemma}
The above result is crucial in that it establishes that the cumulative ratios $\alpha_t$ can be interpreted as Riesz representers for a particular set of parameters. This interpretation leads to a novel estimation strategy for the ratios $\alpha_t$, which we discuss in the next section.

A straightforward corollary of the interpretation of the $\alpha_t$ as a Riesz representer is that it is possible to express the \GATT as a weighted mean, analagous to the IPW expressions of treatment effects such as the ATE.

\begin{theorem}
\label{theorem:ipw-formulation}
Under \ref{assumption:positivity}, \ref{assumption:strong-sequential-randomization}, and \ref{assumption:piecewise-invertibility} the \GATT parameter can be written as
\begin{align}
    \theta = \E\left[ \alpha_\tau(A_\tau, H_\tau) Y \right]. 
\end{align} 
\end{theorem}
This result follows directly from Lemma~\ref{lemma:riesz} by noticing that $\theta = \Psi_\tau(m_{\tau})$. A simple weighted estimator of $\theta$ can then be formed by taking the mean of the observed outcomes weighted by the estimated cumulative density ratios $\alpha_{\tau}$. In the single time point case for the ATT, the estimator simplifies to the below.
\setcounter{example}{1}
\begin{example}[continued]
    The weighted representation of the ATT parameter is given by
    \begin{align}
        \theta = \frac{1}{P(A = 1)} \E\left[\frac{\mathbb{I}[A = 0]P(A = 1 | L)}{P(A = 0| L)} Y \right].
    \end{align}
\end{example}

Although the parameters $\alpha_t$ can be estimated relatively easily in the above example by plugging in estimates of the propensity score obtained with regression methods, they are not so easy to estimate in general cases. A first option is to estimate directly the densities $\g_t(a_t\mid h_t)$ and plug that estimate into the definition of $\alpha_t$ \eqref{eq:alpha_t}. While such an approach may be feasible for categorical $A_t$ and for few time points, implementing a similar approach for continuous or multivariate $A_t$ would require estimation of conditional densities and computation of numerical integrals which may be challenging and computationally intensive. As an alternative, we estimate the cumulated ratios directly via empirical loss minimization, drawing on the interpretation derived above of the $\alpha_t$ as Riesz representers \cite{chernozhukov2021automatic, chernozhukov2023automatic}.

The overall goal is to express $\alpha_t$ as the minimizer of a carefully chosen loss function. To set up the optimization problem, express $\alpha_t$ as the solution to an optimization problem over a candidate space $\tilde{\mathcal{A}}$, and manipulating the expression:
\begin{align*}
    \alpha_t &= \argmin_{\tilde \alpha \in \tilde{\mathcal{A}}}\E\{\tilde\alpha(A_t,H_t) - \alpha_t(A_t,H_t)\}^2\\
    &= \argmin_{\tilde \alpha \in \tilde{\mathcal{A}}}\E\{\tilde\alpha(A_t,H_t)^2 - 2\tilde\alpha(A_t,H_t)\alpha_t(A_t,H_t)\}\\
    &= \argmin_{\tilde \alpha \in \tilde{\mathcal{A}}}\{\E\left[\tilde\alpha(A_t,H_t)^2\right] - 2\Psi_{t-1}(\tilde\alpha)\}.
\end{align*}
The critical step is applying \eqref{eq:riesz-representer} in the second line such that the the optimization problem is only in terms of observable quantities.
Further simplification (see Appendix~A5
for details) yields
\begin{align}
    \alpha_t =\argmin_{\tilde \alpha \in \tilde{\mathcal{A}}}\E\left\{\tilde\alpha(A_t,H_t)^2 - \alpha_{t-1}(A_{t-1}, H_{t-1}) \frac{\one\{\underline{A}_{t} \in \underline{\mathcal{B}}_{t}\} }{G_{t-1}(A_{t-1}, H_{t-1}) }b_t(A_t, H_t; \tilde\alpha)\right\}.
\end{align}
To estimate $\alpha_t$ in practice we solve the corresponding empirical minimization problem:
\begin{align}
    \label{eq:empirical-riesz}
    \hat{\alpha}_t =\argmin_{\tilde \alpha \in \tilde{\mathcal{A}}}\Ec_n\left\{\tilde\alpha(A_t,H_t)^2 - \hat{\alpha}_{t-1}(A_{t-1}, H_{t-1}) \frac{\one\{\underline{A}_{t} \in \underline{\mathcal{B}}_{t}\} }{\hat{G}_{t-1}(A_{t-1}, H_{t-1}) }b_t(A_t, H_t; \tilde\alpha)\right\},
\end{align}
where $\hat{\alpha}_0 \equiv 1$ and $\hat{G}_{t-1}(A_{t-1}, H_{t-1})$ is an estimator of $G_{t-1}(A_{t-1}, H_{t-1})$. 
Crucially, estimating $\alpha_t$ in this manner avoids having to separately estimate the numerator and denominator of the probability (density) ratios and subsequently dividing them and mutiplying the ratios cumulatively, as is necessary for a plug-in estimator of the Riesz representer. Rather, using the above loss function allows us to estimate the Riesz representer directly. The choice of candidate space $\mathcal{A}$ implies different options for practically solving the optimization problem. Methods based on flexible splines, random forests, or neural networks allow for rich candidate spaces $\mathcal{A}$. Following \citep{chernuzhukov2022riesznet}, we use a neural network to estimate $\alpha_t$. However, contrary to \citep{chernuzhukov2022riesznet} we chose to use cross-fitting when estimating $\alpha_t$. Without cross-fitting, technical conditions are required to control the complexity of the estimators, such as Donsker \citep{vanderLaanRose11} or critical radius assumptions \citep{chernozhukov2021automatic, wainwright2019}. Cross-fitting is a commonly used method that obviates the need for such assumptions \citep{zheng2011cross, chernozhukov2018dml}, although it potentially requires large amounts of data. With cross-fitting our method can be applied with any empirical loss minimization method used to estimate $\alpha_t$. 

\section{Estimation}
\label{section:estimation}
The analyses in Section~\ref{section:efficiency} reveal that constructing an estimator with desirable properties (efficiency and asymptotic normality) requires ensuring that the bias term $\E_{\P_0}[\D(O; \hat \P)]$ is asymptotically negligible, and assuming that the second-order remainder term $\Rem(\hat{P}, \P_0)$ goes to zero at a certain rate.
Fortunately, estimating the entirety of $\P_0$ is not necessary: it suffices to estimate only the parts of $\P_0$ relevant to $\theta$ and its EIF. In our case, the relevant parts are the regression parameters $m_t$, the cumulative ratios $\alpha_t$, and the conditional probabilities $G_t$. Collectively we refer to these nuisance parameters as $\eta = (m_t, \alpha_t, G_t : t \in \{1, \dots \tau \})$. We first discuss strategies for estimating the nuisance parameters $m_t$, $\alpha_t$, and $G_t$ and then show how they can be used to construct estimators of $\theta$.

The sequential regression parameters $m_t$ can be estimated using flexible regression techniques from the machine learning literature, using the sequential regression approach outlined in Section~\ref{section:causal-effects} and described further below. Similarly, the parameters $G_t$ can be estimated using any binary regression method. The Riesz representers $\alpha_t$ may be estimated using the empirical loss minimization strategy described in Section~\ref{section:riesz}. With strategies in hand for estimating the nuisance parameters $m_t$, $G_t$, and $\alpha_t$, we now turn to constructing estimators of $\theta$. 

\subsection{Substitution and weighting estimators}
A substitution estimator forms an estimate of $\theta$ by plugging the estimates of $m_t$ directly into the identification result \eqref{eq:m_t}. This can be described algorithmically as follows:
\begin{enumerate}
    \item Initialize $\hat{m}_{\tau + 1, i}(A_{\tau + 1, i}^d, H_{\tau + 1, i}) = Y_i$. 
    \item For $t = \tau, \dots, 1$:
    \begin{itemize}
        \item Using any regression method, regress $\hat{m}_{t + 1}(A^d_{t + 1, i}, H_{t + 1}, i)$ on $(A_{t,i}, H_{t,i})$ for all $i \in \{1, \dots, n \}$ such that $\underline{A}_{t + 1, i} \in \underline{\mathcal{B}}_{t + 1}$.
        \item Use the regression model to form predictions $\hat{m}_{t}(A_{t, i}^d, H_{t, i})$ for all $i \in \{1, \dots, n \}$. 
    \end{itemize}
    \item Form the substitution estimator as
    \begin{align}
        \hat{\theta}_{sub} = \frac{1}{\sum_{i=1}^n \1[\bar{A}_i \in \bar{\mathcal{B}}]} \sum_{i=1}^n \1[\bar{A}_i \in \bar{\mathcal{B}}] \hat{m}_{1}(A_{1,i}^d, L_{1, i}). 
    \end{align}
\end{enumerate}
For a weighting estimator, suppose that that we have an estimator $\hat{\alpha}_\tau$ of $\alpha_\tau$. This estimate may be formed for example via \eqref{eq:empirical-riesz}. A weighted estimator of $\theta_1$ is then given by 
\begin{align}
    \hat{\theta}_{ipw} = \frac{1}{n} \sum_{i=1}^n \hat{\alpha}_{\tau}(A_{\tau,i}, H_{\tau,i}) Y_i.
\end{align}
The consistency of each approach depends on the consistency of the underlying nuisance estimators. The substitution estimator will be consistent if all of the regressions are consistent. The weighted estimator is consistent if $\hat{\alpha}_\tau$ is estimated consistently. The sampling distributions of each estimator are generally unknown, save for the case of the nuisance parameters being estimated using well-specified parametric models, which is unlikely to be true in practical scenarios. Thus, we turn next to an estimator that is consistent and asymptotically normal even when data-adaptive nuisance estimators are used.

\subsection{Targeted minimum loss-based estimator}
Targeted Minimum Loss-Based Estimation (TMLE) is a framework for constructing asymptotically efficient non-parametric 
 plug-in estimators \citep{vanderLaanRose11}. The core idea is to solve the EIF estimating equation by carefully fluctuating initial estimates of the parts of $\P_0$ relevant to the parameter of interest. The TMLE estimator for \GATTs is similar to the one given in \cite{diaz2023lmtp} for LMTPs with modifications to incorporate the conditional structure of the \GATT parameter. Following their work, we use cross-fitting in order to avoid technical conditions on the complexity of the nuisance estimators \citep{klaassen1987consistent, zheng2011cross}. Randomly partition the indexes $\{1, \dots, n \}$ into $J$ validation sets $\mathcal{V}_1, \dots, \mathcal{V}_J$. For each $j \in \{1, \dots, J \}$, the training sample is given by $\mathcal{J}_j = \{ 1, \dots, n \} \backslash \mathcal{V}_j$. Let $j(i)$ be the validation set containing index $i$. 

The goal of the TMLE algorithm is to form a set of updated estimates $\tilde{m}_{t,j(i)}$ such that the empirical EIF estimating equation is solved. Inspection of the form of the EIF shows that this will be the case when 
\begin{align}
    \Ec_n\left[ \sum_{t=1}^\tau \hat{\alpha}_t(A_t, H_t) \frac{\1[\underline{A}_{t+1} \in \underline{\mathcal{B}}_{t+1}]}{\hat{G}_t(A_t, H_t)} \left\{ \tilde{m}_{t+1}(A_{t+1}^d, H_{t+1}) - \tilde{m}_{t}(A_t, H_t) \right\} \right] \approx 0. \label{eq:empirical-eif}
\end{align}
TMLE ensures this is the case by fluctuating an initial estimate $\hat{m}_t$ via a carefully chosen parametric submodel indexed by a parameter $\epsilon \in \mathbb{R}$ and loss function to form updated estimates $\tilde{m}_t$. The parametric submodel is chosen such that the gradient of the loss with respect to $\epsilon$ equals \eqref{eq:empirical-eif}. When the parameter $\epsilon$ is estimated by minimizing the loss function, this ensures that at the minimizer the empirical mean of the gradient is approximately zero. As such, the updated estimates approximately solve the empirical estimating equation \eqref{eq:empirical-eif}. For \GATTs, the fluctuation model is chosen to be a weighted generalized linear model with canonical link, an intercept parameter $\epsilon$, and offset set to the initial estimates $\hat{m}_t$. 
The TMLE algorithm is as follows:
\begin{enumerate}
    \item Initialize $\tilde{\eta} = \hat{\eta}$ and $\tilde{m}_{\tau + 1, j(i)}(A^d_{\tau + 1, i}, H_{\tau + 1, i}) = Y_i$. 
    \item For $s = 1, \dots, \tau$ compute weights 
    \begin{align}
        \omega_{s,i} = \frac{\1[\underline{A}_{t+1} \in \underline{\mathcal{B}}_{t+1}]}{\hat{G}_t(A_{t,i}, H_{t,i})} \hat{\alpha}_{s, i}(A_{t,i}, H_{t,i}).
    \end{align}
    \item For $t = \tau, \dots, 1$:
    \begin{itemize}
        \item Find the maximum likelihood estimate $\hat{\epsilon}$ of $\epsilon$ under the model
        \begin{align}
            \mathrm{link}(\tilde{m}_t^\epsilon(A^d_{t,i}, H_{t,i}) = \epsilon + \mathrm{link}( \tilde{m}_{t,j(i)}(A_{t,i}, H_{t,i})). 
        \end{align}
        with weights computed in Step 2.  
        \item Update $\tilde{m}$ as
        \begin{align}
            \mathrm{link}(\tilde{m}_t^{\hat{\epsilon}}(A_{t,i}, H_{t,i}) &= \hat{\epsilon} + \mathrm{link}( \tilde{m}_{t,j(i)}(A_{t,i}, H_{t,i})), \\
            \mathrm{link}(\tilde{m}_t^{\hat{\epsilon}}(A^d_{t,i}, H_{t,i}) &= \hat{\epsilon} + \mathrm{link}( \tilde{m}_{t,j(i)}(A^d_{t,i}, H_{t,i})). 
        \end{align}
    \end{itemize}
    \item The TMLE estimate is 
    \begin{align}
        \hat{\theta}_{tmle} = \frac{1}{\sum_{i=1}^n \1\left[ \bar{A}_i \in \bar{\mathcal{B}} \right]} \sum_{i=1}^n \1\left[ \bar{A}_i \in \bar{\mathcal{B}} \right] \tilde{m}_{1, j(i)}(A_{1,i}^d, L_{1,i}).
    \end{align}
\end{enumerate}
The TMLE algorithm for \GATT parameters is similar to the TMLE for LMTPs \citep{diaz2023lmtp}. The key difference between the algorithms is in Step 2, where for \GATTs the weights are multiplied by an indicator of the future treatment assignment lying in the conditioning set. The statistical properties of TMLE flow from the fact that after this iterative procedure, \eqref{eq:empirical-eif} is satisfied. Using this, it is possible to show that the TMLE estimator is asymptotically normal and efficient.
\begin{theorem}
    \label{theorem:tmle-weak-convergence}
    Assume that, for each $j \in \{1, \dots, J \}$, 
    \begin{align}
         \sum_{t=1}^\tau \left\| \hat{\alpha}_{t,j} - \alpha_t \right\| \|\tilde{m}_{t,j} - m_t\| = o_P(n^{-1/2}).
    \end{align}
    and 
    \begin{align}
         \sum_{t=1}^\tau \left\| \hat{G}_{t,j} - G_t \right\| \|\tilde{m}_{t,j} - m_t\| = o_P(n^{-1/2}).
    \end{align}
    Assume there exists some $c < \infty$ such that $\P(\alpha_t < c) =1$ and $\P(\hat{\alpha}_t(A_t, H_t) < c) = 1$. 
    Then
    \begin{align}
        \sqrt{n}(\hat{\theta}_{tmle} - \theta) \leadsto N(0, \sigma^2), 
    \end{align}
    where $\sigma^2 = \mathrm{Var}_{\P_0}(\D(Z; \P_0))$. 
\end{theorem}
Satisfying the assumptions of Theorem~\ref{theorem:tmle-weak-convergence} requires that all the nuisance parameters are estimated at sufficiently fast rates. Consistency of $\hat{\theta}_{tmle}$, however, can be achieved even if some of the nuisance parameters are not estimated consistently, as shown in the following theorem.
\begin{theorem}[$\tau + 1$ multiply robust consistency of TMLE]
    \label{theorem:robustness}
    Assume that $\|\hat{G}_t - G_t \| = o_p(1)$ for all $t \in \{1, \dots, \tau - 1 \}$ and that there exists a $k \in \{1, \dots, \tau - 1 \}$ such that $\| \tilde{m}_t - m_t \| = o_p(1)$ for all $t > k$ and $\| \hat{\alpha}_t - \alpha_t \| = o_p(1)$ for all $1 < t \leq k$. Then $\hat{\theta}_{tmle} - \theta = o_p(1)$. 
\end{theorem}

It is also theoretically possible to construct estimator with stronger multiple robustness properties, similar to the sequentially double robust (SDR) estimator proposed for LMTPs \citep{diaz2023lmtp}. This estimator would require cumulative density ratios of the form $\prod_{t=s}^k r_{s}(A_s, H_s)$ for all combinations of $1 \leq t < \tau$ and $t < s \leq \tau$. \cite{diaz2023lmtp} estimated these ratios by plugging in estimates of $r_t$ for $t \in \{1, \dots, \tau \}$. While it is possible to use the Riesz representer approach to estimate these cumulative ratios directly and develop an SDR estimator, computation costs may be prohibitive for large $\tau$ as it would require solving \eqref{eq:empirical-riesz} a total of $\tau(\tau - 1) / 2$ times. In comparison, for the TMLE algorithm developed, we require solving \eqref{eq:empirical-riesz} only $\tau$ times.

\section{Simulation studies}
\label{section:simulation}

We investigate the finite-sample performance of the proposed estimators for the \GATT parameter through two simulation studies, each probing the estimators from a different angle. Simulation study 1 looks at the robustness of the TMLE estimator under inconsistent nuisance parameter estimation. Simulation study 2 compares how the Riesz representer estimation strategy compares to estimating individual density ratios. 
\ifx\anonymized\undefined
Reproduction materials for the simulation studies are available at \url{https://github.com/herbps10/gatt_paper}, drawing on an implementation of our proposed methods in an R package available at \url{https://github.com/nt-williams/lmtp/tree/riesz} \citep{R}.
\else
Reproduction materials for the simulation studies are available at (anonymized), drawing on an implementation of our proposed methods in an R package available at (anonymized) \citep{R}.
\fi

\subsection{Simulation study 1}
\label{section:simulation-study-1}
In the first simulation study we investigate the performance of the TMLE estimator for a \GATT parameter defined for a categorical treatment. The setup is adapted from the simulation study presented in \cite{diaz2023lmtp}. The data generating process is given by
\begin{align}
    L_1 &\sim \mathrm{Categorical}(0.5, 0.25, 0.25), \\
    A_1 | L_1 &\sim \mathrm{Binomial}(5, \mathrm{logit}^{-1}(-0.3 L_1)), \\
    L_t | (\bar{A}_{t-1},\bar{L}_{t-1}) &\sim \mathrm{Bernoulli}(\mathrm{logit}^{-1}(-0.3 L_{t-1} + 0.5 A_{t-1})) \text{ for } t\in\{2, 3, 4\}, \\
    A_t | (\bar{A}_{t-1}, \bar{L}_{t}) &\sim \mathrm{Binomial}(5, \mathrm{logit}^{-1}(-2.5 + A_{t-1} + 0.5 L_{t})) \text{ for } t \in \{ 2, 3, 4 \}, \\
    Y | (\bar{A}_4, \bar{L}_4) &\sim \mathrm{Bernoulli}(\mathrm{logit}^{-1}(-1 + 0.5 A_4 - L_4 )).
\end{align}
The MTP is defined as
\begin{align}
    d(a_t, h_t) = \begin{cases}
        a_t - 1 & \text{if } a_t \geq 1, \\
        a_t     & \text{if } a_t < 1. \\
    \end{cases}
\end{align}
We define two \GATT parameters conditional on the final treatment assignment. That is, we set $\mathcal{B}_4^a = \{ a \}$ for $a \in \{ 0, 1, \dots, 5 \}$. We created $200$ datasets for each sample size   $N = \{ 500, 1000, 2000 \}$ by independently sampling from the data generating process. The nuisance estimators for $m_t$ and $\alpha_t$ were estimated differently in four scenarios:
\begin{enumerate}
    \item $m_t$ and $\alpha_t$ estimated consistently.
    \item $m_t$ estimated consistently for $t > 2$ and inconsistently otherwise; $\alpha_t$ consistently for $t \leq 2$ and inconsistently otherwise.
    \item $m_t$ estimated consistently for $t < 4$ and inconsistently for $t = 4$; $\alpha_t$ estimated consistently for $t = 4$ and inconsistently otherwise.
    \item $m_t$ and $\alpha_t$ estimated inconsistently.
\end{enumerate}
Consistent estimation of $m_t$ was achieved by Super Learning with \texttt{SL.mean} and \texttt{SL.glm} learners, where the generalized linear model was correctly specified. For $\alpha_t$, we used a linear model including all first-order interaction terms (details are provided in Appendix~A7
). For the inconsistent cases, only \texttt{SL.mean} was used for estimating $m_t$ and $\alpha_t$ was fixed to one. In all cases the conditional probability of inclusion in the conditioning set (the nuisance parameter $G_t$) was estimated with Super Learner with \texttt{SL.mean} and \texttt{SL.ranger} learners.

A subset of the results covering parameter estimates when the conditioning set is $\mathcal{B}_4^a = \{ 1 \}$ are shown in Table~\ref{tab:simulation-results-1}; results for all other conditioning sets are available in Appendix~A8
. The TMLE estimator achieved near zero mean error and mean absolute error for nearly all combinations of sample size and conditioning sets in scenario (1), where all nuisance parameters are estimated consistently. Empirical coverage was near optimal at the largest sample size, and conservative at smaller sample sizes. In scenario (2), IPW appears to be inconsistent, as expected because $\alpha_4$ is estimated inconsistently. As expected based on the $\tau + 1$ robustness result (Theorem~\ref{theorem:robustness}), TMLE has good performance in terms of error metrics. Surprisingly, the substitution estimator also has good performance in this example. In scenario (3), TMLE also has good performance both in terms of empirical coverage and error metrics, which is unexpected as the scenario does not fulfill the requirements of Theorem~\ref{theorem:robustness}. We hypothesize that estimating $\alpha_4$ correctly, which is the cumulative densities from $t = 1$ to $t=4$, allows TMLE to correct for the inconsistent estimation of $m_4$.   Finally, in scenario (4) all estimators have poor performance due to the inconsistent estimation of both nuisance parameters.

\ifx\figurestables\undefined 
\else
\begin{table}[ht]
    \centering
    \begin{tabular}{|rrrrrrrrrr|}
        \hline
        & & & 95\% Coverage & \multicolumn{3}{c}{MAE $\times$ 100} & \multicolumn{3}{c|}{ME $\times$ 100} \\
        $\mathcal{B}^a$ & Scenario & $N$ & TMLE & IPW &  Sub & TMLE & IPW & Sub & TMLE  \\
        \hline
        \{ 1 \} & 1 & 500 & 97.5\% & 18.66 & 4.53 & 6.39 & -17.59 & -0.96 & 0.03\\
         &  & 1000 & 97.5\% & 6.32 & 3.07 & 4.32 & -4.23 & -0.01 & 0.66\\
         &  & 2500 & 97.0\% & 3.16 & 1.82 & 2.47 & -2.04 & -0.01 & -0.05\\
         &  & 5000 & 94.0\% & 2.04 & 1.41 & 1.92 & -1.12 & -0.06 & -0.11\\
         & 2 & 500 & 74.5\% & 15.68 & 4.41 & 4.60 & -15.68 & -1.43 & -1.19\\
         &  & 1000 & 76.0\% & 15.95 & 2.91 & 3.23 & -15.95 & -0.09 & -0.23\\
         &  & 2500 & 77.0\% & 15.72 & 1.79 & 1.93 & -15.72 & 0.12 & 0.01\\
         &  & 5000 & 74.5\% & 15.77 & 1.35 & 1.52 & -15.77 & 0.04 & -0.04\\
         & 3 & 500 & 74.5\% & 24.36 & 15.72 & 7.60 & -9.63 & -15.72 & -2.94\\
         &  & 1000 & 76.5\% & 14.42 & 15.97 & 7.34 & -1.74 & -15.97 & -4.42\\
         &  & 2500 & 62.0\% & 14.76 & 15.73 & 6.60 & -2.87 & -15.73 & -4.75\\
         &  & 5000 & 53.5\% & 11.85 & 15.77 & 5.63 & -1.89 & -15.77 & -4.79\\
         & 4 & 500 & 0.0\% & 15.68 & 15.73 & 15.72 & -15.68 & -15.73 & -15.72\\
         &  & 1000 & 0.0\% & 15.95 & 15.96 & 15.96 & -15.95 & -15.96 & -15.96\\
         &  & 2500 & 0.0\% & 15.72 & 15.73 & 15.73 & -15.72 & -15.73 & -15.73\\
         &  & 5000 & 0.0\% & 15.77 & 15.77 & 15.77 & -15.77 & -15.77 & -15.77\\
        \hline
    \end{tabular}
    \caption{Results of Simulation Study 1 showing empirical coverage of the 95\% confidence intervals, Mean Absolute Error (MAE), and Mean Error (ME) for the inverse probability weighted estimator (IPW), substitution estimator (Sub), and Targeted minimum loss-based estimator (TMLE).}
    \label{tab:simulation-results-1}
\end{table}
\fi

\subsection{Simulation study 2}
The second simulation study investigates how the Riesz representation approach for estimating the cumulated densities $\alpha_t$ compares to estimating each ratio $r_t$ and then cumulating them to form estimates of $\alpha_t$. As such, we consider a relatively simple longitudinal data generating process and modified treatment policy because the main variable of interest in this simulation study is the number of time points.
The data generating process is given by
\begin{align}
    L_1 &\sim \mathrm{Uniform}(0, 1), \\
    A_1 | L_1 &\sim \mathrm{Bernoulli}(0.5), \\
    L_t | (\bar{A}_{t-1},\bar{L}_{t-1}) &\sim \mathrm{Normal}(0.25L_{t-1} , 0.5^2) \text{ for } t\in\{2, \dots, \tau\}, \\
    A_t | L_t &\sim \mathrm{Bernoulli}(\mathrm{logit}^{-1}(0.5 + 0.1 L_t)), \\
    Y | A_\tau, L_\tau &\sim \mathrm{Normal}(A_\tau + L_\tau, \sigma^2).
\end{align}
The modified treatment policy was defined as $d(a_t, h_t) = 1$. The LMTP parameter was not conditioned on treatment (or, equivalently, $\mathcal{B}_t = \{0, 1\}$ for all $t$). 
We applied correctly specified generalized linear models to estimate the outcome regression models. We estimated $r_t$ based on correctly specified logistic regressions which were then cumulated to form estimates of $\alpha_t$ (we refer to this as the ``TMLE plug-in" method,
as estimates $\hat{r}_t$ are plugged in to the definition of $\alpha_t$).

Simulation results are shown in Table~\ref{tab:simulation-results-2} and Figure~\ref{fig:simulation-results-2}. For sample sizes $N = 500$ and $N = 1000$, the TMLE RR estimator maintains near optimal empirical coverage for all $\tau$, while the performance of the TMLE Plug-in estimator suffers as $\tau$ increases. A similar pattern is seen for the mean absolute error. The instability of the Plug-in estimator can be seen via the standard deviation of the estimated cumulative weights $\hat{\alpha}_\tau$ which are larger for the TMLE plug-in estimator than the TMLE RR estimator.

\ifx\figurestables\undefined 
\else
\begin{table}[ht]
    \centering
    \begin{tabular}{|rrrrrrrr|}
        \hline
        & & \multicolumn{2}{c}{95\% Coverage} & \multicolumn{2}{c}{MAE $\times$ 100} & \multicolumn{2}{c|}{$\mathrm{sd}(\hat{\alpha}_\tau)$} \\
        $N$ & $\tau$ & TMLE RR & TMLE Plug-in & TMLE RR & TMLE Plug-in & TMLE RR & TMLE Plug-in \\
        \hline
        500 & 2 & 96.5\% & 97.0\% & 3.75 & 3.85 & 1.42 & 1.59\\
         & 4 & 93.0\% & 93.5\% & 5.88 & 7.15 & 2.21 & 3.20\\
         & 6 & 93.0\% & 94.5\% & 6.97 & 13.73 & 2.20 & 6.16\\
         & 8 & 89.5\% & 95.0\% & 7.65 & 22.88 & 2.62 & 11.12\\
         & 10 & 92.5\% & 87.0\% & 7.81 & 31.48 & 2.84 & 18.91\\
         & 12 & 96.5\% & 71.5\% & 7.85 & 31.68 & 3.48 & 28.06\\
         & 14 & 98.5\% & 47.0\% & 9.16 & 36.39 & 4.52 & 30.75\\
        1000 & 2 & 91.5\% & 93.5\% & 3.21 & 3.26 & 1.41 & 1.57\\
         & 4 & 94.5\% & 95.0\% & 4.10 & 4.94 & 2.25 & 3.12\\
         & 6 & 88.5\% & 95.0\% & 5.46 & 9.13 & 2.20 & 5.78\\
         & 8 & 92.5\% & 93.5\% & 5.40 & 16.35 & 2.49 & 10.32\\
         & 10 & 88.5\% & 87.0\% & 6.25 & 27.95 & 2.52 & 18.21\\
         & 12 & 93.5\% & 85.5\% & 6.00 & 40.62 & 2.87 & 28.99\\
         & 14 & 96.0\% & 58.5\% & 6.21 & 35.31 & 3.62 & 34.89\\
        2000 & 2 & 95.5\% & 97.0\% & 2.04 & 2.08 & 1.40 & 1.55\\
         & 4 & 95.0\% & 95.0\% & 3.11 & 3.59 & 2.29 & 3.07\\
         & 6 & 94.0\% & 96.5\% & 3.49 & 6.16 & 2.24 & 5.61\\
         & 8 & 90.5\% & 96.5\% & 4.01 & 9.96 & 2.46 & 9.89\\
         & 10 & 91.5\% & 90.5\% & 3.86 & 21.79 & 2.41 & 17.47\\
         & 12 & 90.5\% & 82.5\% & 4.18 & 31.81 & 2.59 & 29.09\\
         & 14 & 93.0\% & 74.0\% & 4.93 & 33.63 & 3.07 & 39.81\\
        \hline
    \end{tabular}
    \caption{Results of Simulation Study 2 comparing the performance of the TMLE estimator with Riesz representers $\alpha_t$ estimated using empirical loss minimization (TMLE RR) and via plug-in estimation (TMLE Plug-in). The estimators are compared by their empirical 95\% confidence interval coverage, mean absolute error (MAE), and the mean standard deviation of the Riesz representers at time $\tau$ ($\mathrm{sd}(\hat{\alpha}_\tau)$).}
    \label{tab:simulation-results-2}
\end{table}

\begin{figure}
    \centering
    \includegraphics[width=\columnwidth]{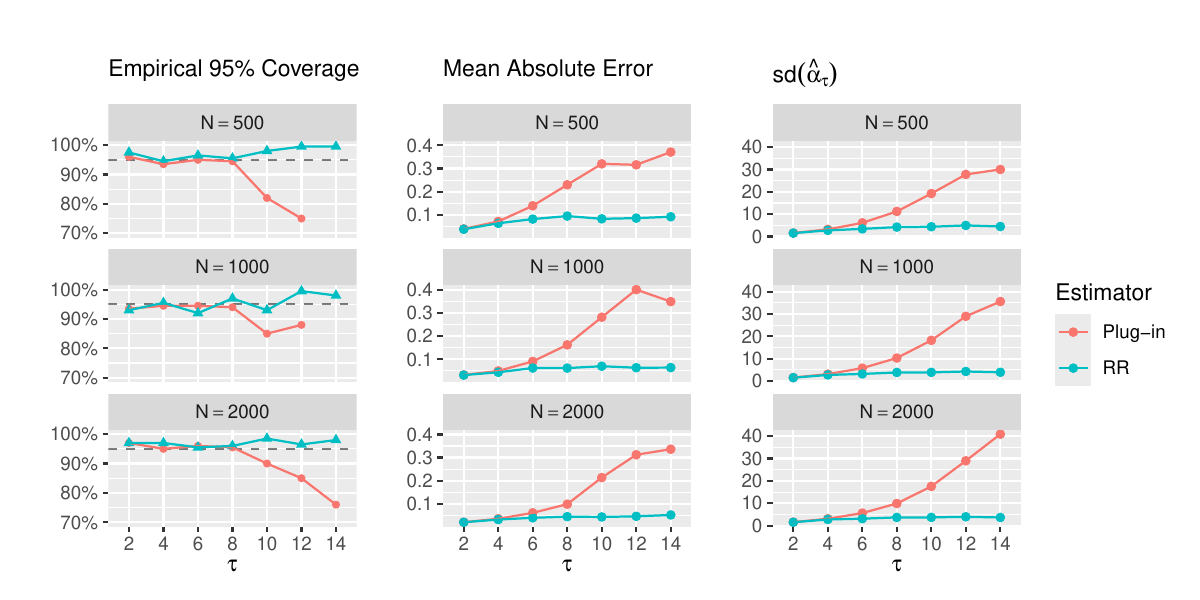}
    \caption{Results of Simulation Study 2.}
    \label{fig:simulation-results-2}
\end{figure}

\fi

\section{Application}
\label{section:application}

We revisit an application recently considered by \citet{rudolph2024associations}. In this recent work, the authors considered a set of nine chronic pain management strategies $A = (A_1 ,A_2, \dots, A_9)$ and were interested in the extent to which hypothetically intervening on each (as opposed to no intervention) would increase risk of developing opioid use disorder (denoted $Y$) over the subsequent 12 months among Medicaid beneficiaries with chronic pain. For this example, the authors assumed observed data $O = (W, A, C, CY)$, where $W$ denotes baseline covariates measured during the initial 6-months of Medicaid enrollment that served as a washout period, $A$ were the treatments measured during the subsequent 6 months, $C$ indicates whether the patient was still enrolled by month 24 when the outcome was measured, and $Y$ is the outcome of a new diagnosis of opioid use disorder (OUD) by 24 months post-enrollment, which is observed among those who remain enrolled.

Two out of these nine treatments are features of opioid prescribing: maximum daily opioid prescription dose (measured in morphine milligram equivalents (MME), denoted $A_1$) and proportion of days covered by an opioid prescription out of the 6-month period (denoted $A_2$). These two treatments are of policy relevance in terms of prescribing guidelines, published by both the Centers for Disease Control and Prevention and Veterans Affairs \citep{dowell2022cdc,rosenberg2018opioid}. For each of these two continuous variables, the authors considered a hypothetical intervention that was a multiplicative shift that increased the observed treatment variable by 20\%, leaving the remaining treatment variables at their observed values, denoted $d_1(a) = (1.2\times a_1, a_2, \dots, a_9)$ and 2) $d_2(a) = (a_1, 1.2\times a_2, \dots, a_9)$, and setting everyone to uncensored (denoted $C=1$). \citet{rudolph2024associations} then estimated the effect of each of these two interventions among the entire population of chronic pain patients enrolled in Medicaid, $\E[Y(A^{d_k},C=1)] / \E[Y(A,C=1)]$ for $k \in \{1,2\}$. 

However, only 40\% of this cohort of chronic pain patients reported any opioid pain prescription during the 6 month treatment period \citep{rudolph2024associations}. The intervention $d_k$ would keep these individuals at their observed value of 0 for each of the treatments, and they would contribute a 0 to the overall effect. 
It may be more policy relevant to estimate the effect of shifting these variables only among those individuals who actually received an opioid prescription, as this would be the population impacted by changes in opioid prescribing guidelines to reduce dose and duration. 

Estimating the effect of shifts in opioid dose and duration among those patients who received an opioid is an example of a \GATT. We can denote this estimand: $\E[Y(A^{d_k},C=1) \mid A_k>0] / \E[Y(A,C=1) \mid A_k>0]$ for $k \in \{1,2\}$. We use the TMLE estimator proposed in Section 4.2 to estimate the above two quantities; the results are reported in Table \ref{tab:app}. Outcome regression nuisance parameters were estimated using the super learner algorithm, with candidates including an intercept-only model, a main-effects GLM, boosting, random forests, a single-layer neural network, and multivariate adaptive regression splines. Riesz representers were estimated with an L2-regularized multilayer perceptron using the Adam optimizer. As in the previous work \citep{rudolph2024associations}, we find that higher doses 
increase risk of incident OUD. Specifically, increasing dose by 20\% would be expected to increase risk of OUD by 4.9\% (4.4-5.5\%), among chronic pain patients with an opioid prescription.  

\ifx\figurestables\undefined
\else
\begin{table}[ht]
\centering
\caption{Estimated effects of increasing each of opioid dose and duration (vs. leaving them at their observed levels) on risk of incident opioid use disorder among chronic pain patients in Medicaid. Point estimates and and 95\% CIs.} 
\label{tab:app}
\begin{tabular}[t]{p{.6cm}cccccc}
\toprule
\multicolumn{1}{c}{ } & \multicolumn{2}{c}{$\E[Y(A^{d_k},C=1) \mid A_k>0]$} & \multicolumn{2}{c}{$\E[Y(A,C=1) \mid A_k>0]$.} & \multicolumn{2}{c}{Risk Ratio} \\
 & Est. & 95\% CI & Est. & 95\% CI & Est. & 95\% CI\\
\midrule
$d_1$ & 0.044 & 0.043, 0.045 & 0.042 & 0.041, 0.043 & 1.049 & 1.044, 1.055\\
$d_2$ & 0.036 & 0.035, 0.037 & 0.036 & 0.035, 0.037 & 1.001 & 0.999, 1.004\\
\bottomrule
\end{tabular}
\end{table}
\fi

\section{Discussion}
\label{section:discussion}
In this work we introduced \GATTs, which generalize the ATT to multi-valued or continuous treatments, longitudinal data structures, and complex interventions posed in the form of MTPs. The \GATT is identifiable in the form of sequential regressions under a familiar strong sequential randomization assumption, and a positivity assumption that, echoing the case for the single time point ATT, is weaker than that required for comparable population longitudinal parameters. 

It is notable that in order for the \GATT parameter to be identifiable, it is defined conditional on the longitudinal natural value of treatment falling in a set, formally the event $\bar{A}(d) \in \bar{\mathcal{B}}$. An alternative, perhaps more obvious generalization of the ATT would be to condition on the observed treatment trajectory falling in the conditioning set: $\bar{A} \in \bar{\mathcal{B}}$. While conditioning on the observed longitudinal treatments is more intuitive, we found that doing so leads to a causal parameter that is not in general causally identifiable. Fortunately, Condition~\ref{condition:comparable} provides a framework for constructing parameters that are both interpretable and identifiable that are also of wide scientific relevance.

Using tools from semi-parametric efficiency theory, we derived the form of the EIF of the \GATT, which defines its semi-parametric efficiency bound. Interestingly, the EIF involves the nuisance parameter $G_t$, which is the conditional probability at each time point of future observations falling within the conditioning set conditional on the present and past. This leads to estimators that have more complex multiple-robustness properties than estimators for unconditional LMTP parameters. Specifically, Theorem~\ref{theorem:robustness} shows that in order for the proposed TMLE estimator to be consistent, $G_t$ must be estimated consistently at all time points, while the outcome regression $m_t$ can be inconsistent after a certain time point so long as the Riesz representer $\alpha_t$ is subsequently consistently estimated. Thus, it is critical that $G_t$ be estimated using data-adaptive algorithm or via an ensemble algorithms to have a better chance of consistent estimation. 

A key element of our approach lies in estimating the cumulated ratios $\alpha_t$ using empirical loss minimization based on interpreting $\alpha_t$ as a Riesz representer. The results of Simulation Study 2 show that this approach has empirical benefits over estimating the conditional treatment probabilities (or densities) directly and then plugging them in to the expression for $\alpha_t$ \eqref{eq:alpha_t}.  While both methods of estimating Riesz representers appear to have yielded a consistent estimator of the \GATT, as expected theoretically, the estimates based on Riesz loss minimization were more stable and converged faster. 
There appears to be a careful balance required in choosing the Riesz representer candidate space $\tilde{\mathcal{A}}$ for the minimization problem: it should be large enough to contain the true Riesz representer functional, but not be so large as to allow overfitting. 
For this reason we relied on ensemble learning, which chooses between multiple candidate estimators of $\alpha_t$ based on cross-validation, in order to avoid overfitting. We focused primarily on deep neural networks and linear models as underlying Riesz learners. Tree-based methods have also been proposed \citep{chernuzhukov2022riesznet}, and we look forward to future work in this area as the ultimate utility of this approach depends on the availability of practical data-adaptive algorithms. 

\ifx\anonymized\undefined 
    \subsection*{Acknowledgements}
    Iv\'an D\'iaz and Kara Rudolph were supported through a Patient-Centered Outcomes Research Institute (PCORI) Project Program FundingAward (ME-2021C2-23636-IC). The computational requirements for this work were supported in part by the NYU Langone High Performance Computing (HPC) Core's resources and personnel.
\fi

\ifx\arxiv\undefined 
\else
\appendix
\section{Appendix}
\label{section:appendix}
    \subsection{Proof of Theorem 1}
\label{appendix:theorem-1}
Recall the definition of the longitudinal natural value of treatment: $\bar A(d) = (A_1, A_2(d_1), ..., A_\tau(d_{\tau-1}))$. Begin by rewriting the parameter as
{\footnotesize
\begin{align*}
    & \Ec\left[ Y(\bar{A}^d) \mid \bar{A}(d) \in \bar{\mathcal{B}} \right] = \Ec\left[ \frac{\one\{A_1 \in \mathcal{B}_1\}}{P(A_1 \in \mathcal{B}_1\mid \underline{A}_2(d) \in \mathcal{\underline{B}}_2)}Y(\bar{A}^d) \mid \underline{A}_2(d) \in \mathcal{\underline{B}}_2 \right].  \\
    \intertext{Next, apply the law of iterated expectations and simplify:}
    &= \int_{\mathcal{A}_1,\mathcal{L}_1}  \Ec\left[\frac{\one\{A_1 \in \mathcal{B}_1\}}{P(A_1 \in \mathcal{B}_1\mid \underline{A}_2(d) \in \mathcal{\underline{B}}_2)} Z_1(a_1^d, l_1) \mid A_1 = a_1, L_1 = l_1, \underline{A}_2(d)\in \mathcal{\underline{B}}_2 \right] d\P(a_1, l_1\mid \underline{A}_2(d) \in \mathcal{\underline{B}}_2) \\
    &= \int_{\mathcal{A}_1,\mathcal{L}_1} \one\{a_1 \in \mathcal{B}_1\}\frac{P(\underline{A}_2(d) \in \mathcal{\underline{B}}_2)}{P(\underline{A}_1(d) \in \mathcal{\underline{B}}_1)} \Ec\left[ Z_1(a_1^d, l_1) \mid A_1 = a_1, L_1 = l_1, \underline{A}_2(d)\in \mathcal{\underline{B}}_2 \right] d\P(a_1, l_1\mid \underline{A}_2(d) \in \mathcal{\underline{B}}_2).
    \intertext{Applying $\one\{a_1 \in \mathcal{B}_1\} = P(A_1 \in \mathcal{B}_1\mid A_1=a_1, L_1=l_1, \underline{A}_2(d) \in \mathcal{\underline{B}}_2)$, we get}
    &= \int_{\mathcal{A}_1,\mathcal{L}_1} \Ec\left[ Z_1(a_1^d, l_1) \mid A_1 = a_1, L_1 = l_1, \underline{A}_2(d)\in \mathcal{\underline{B}}_2 \right] d\P(a_1, l_1\mid \underline{A}_1(d) \in \mathcal{\underline{B}}_1)\\
        &= \int_{\mathcal{A}_1,\mathcal{L}_1} \frac{1}{P( \underline{A}_2(d)\in \mathcal{\underline{B}}_2 \mid A_1=a_1, L_1=l_1)}\Ec\left[ Z_1(a_1^d, l_1)\one[ \underline{A}_2(d)\in \mathcal{\underline{B}}_2 ] \mid A_1 = a_1, L_1 = l_1\right] d\P(a_1, l_1\mid \underline{A}_1(d) \in \mathcal{\underline{B}}_1)
    \intertext{By Lemma~\ref{lemma:conditional-independence} (below) and Assumptions~\ref{assumption:positivity} and \ref{assumption:strong-sequential-randomization},}
        &= \int_{\mathcal{A}_1,\mathcal{L}_1} \frac{1}{P( \underline{A}_2(d)\in \mathcal{\underline{B}}_2 \mid A_1=a_1, L_1=l_1)}\Ec\left[ Z_1(a_1^d, l_1)\one[ \underline{A}_2(d)\in \mathcal{\underline{B}}_2 ] \mid A_1 = a_1^d, L_1 = l_1\right] d\P(a_1, l_1\mid \underline{A}_1(d) \in \mathcal{\underline{B}}_1)\\  
        &= \int_{\mathcal{A}_1,\mathcal{L}_1} \frac{P( \underline{A}_2(d)\in \mathcal{\underline{B}}_2 \mid A_1=a_1^d, L_1=l_1)}{P( \underline{A}_2(d)\in \mathcal{\underline{B}}_2 \mid A_1=a_1, L_1=l_1)}\Ec\left[ Z_1(a_1^d, l_1)\mid A_1 = a_1^d, L_1 = l_1, \underline{A}_2(d)\in \mathcal{\underline{B}}_2\right] d\P(a_1, l_1\mid \underline{A}_1(d) \in \mathcal{\underline{B}}_1)\\  
&= \int_{\mathcal{A}_1,\mathcal{L}_1} \Ec\left[ Z_1(a_1^d, l_1)\mid A_1 = a_1^d, L_1 = l_1, \underline{A}_2(d)\in \mathcal{\underline{B}}_2\right] d\P(a_1, l_1\mid \underline{A}_1(d) \in \mathcal{\underline{B}}_1)
    \intertext{Continue by applying the same steps to the inner expectation:}
    &= \int_{\bar{\mathcal{A}}_2,\bar{\mathcal{L}}_2} \Ec\left[ Z_1(a_1^d, l_1) \mid A_2 = a_2, L_2 = l_2, A_1 = a_1^d, L_1 = l_1, \underline{A}_2(d) \in \underline{\mathcal{B}}_2 \right] \times \\
    &\times d\P(a_2, l_2 \mid A_1=a_1^d, L_1=l_1, \underline{A}_2(d) \in \mathcal{\underline{B}}_2) d\P(a_1, l_1 \mid \underline{A}_1(d) \in \mathcal{B}_1)\notag
    \intertext{In the event $A_{t+1} = a_{t+1}, H_{t+1} = h^d_{t+1}$ then $Z_{t}(\bar{a}_t^d, \bar{l}_t) = Z_{t+1}(\bar{a}_{t+1}^d, \bar{l}_{t+1})$. Therefore:}
    &= \int_{\bar{\mathcal{A}}_2,\bar{\mathcal{L}}_2} \Ec\left[ Z_2(\bar a_2^d, \bar l_2) \mid A_2 = a_2, L_2 = l_2, A_1 = a_1^d, L_1 = l_1, \underline{A}_2(d) \in \underline{\mathcal{B}}_2 \right] \times \\
    &\times d\P(a_2, l_2 \mid A_1=a_1^d, L_1=l_1, \underline{A}_2(d) \in \mathcal{\underline{B}}_2) d\P(a_1, l_1 \mid \underline{A}_1(d) \in \mathcal{B}_1)
    \intertext{Continue applying the previous steps until arriving at}
    &= \int_{\bar{\mathcal{A}}_\tau,\bar{\mathcal{L}}_\tau} \Ec\left[ Z_\tau(\bar{a}_\tau^d, \bar{l}_\tau) \mid A_\tau = a_\tau^d, H_\tau = h_\tau^d, \underline{A}_{\tau + 1}(d) \in \underline{\mathcal{B}}_{\tau + 1} \right] \prod_{k=1}^\tau d\P(a_k, l_k \mid a_{k-1}^d, h_{k-1}^d, \underline{A}_k(d) \in \underline{\mathcal{B}}_k),
    \intertext{Finally, by applying the definition of $Z_t$:}
    &= \int_{\bar{\mathcal{A}}_\tau,\bar{\mathcal{L}}_\tau} \Ec\left[ Y \mid A_\tau = a_\tau^d, H_\tau = h_\tau^d, \underline{A}_{\tau + 1}(d) \in \underline{\mathcal{B}}_{\tau + 1} \right] \prod_{k=1}^\tau d\P(a_k, l_k \mid a_{k-1}^d h_{k-1}^d, \underline{A}_k(d) \in \underline{\mathcal{B}}_k), \label{eq:identification-integral}
\end{align*}}
Next, recursively apply the following starting with $t=\tau$ to arrive at the stated result:
\begin{align}
    \int_{\mathcal{A}_t, \mathcal{L}_t} m_t(a_t^d, h_t^d) d\P(a_t, l_t \mid a_{t-1}^d, h_{t-1}^d, \underline{A}_t \in \underline{\mathcal{B}}_t) &= \Ec\left[ m(A_t^d, H_t) \mid A_{t-1} = a_{t-1}^d, H_{t-1} = h_{t-1}^d, \underline{A}_t \in \underline{\mathcal{B}}_t \right] \\
    &= m_{t-1}(a_{t-1}^d, h_{t-1}^d). 
\end{align}

\begin{lemma}
    \label{lemma:conditional-independence}
    Given Assumption~\ref{assumption:strong-sequential-randomization}, for all $t$ it follows that $Z_t(\overline{a}_t^d, \overline{l}_t) \1[\underline{A}_{t+1}(d) \in \underline{\mathcal{B}}_{t+1}] \indep A_t \mid H_t=h_t^d$. 
\end{lemma}
\begin{proof}
    Under the assumed structural causal model, in the event $H_t=h_t^d$, $Z_t(\bar{a}^d_t, \bar{l}_t) \1[\underline{A}_{t+1}(d) \in \underline{\mathcal{B}}_{t+1}]$ is a deterministic function of $(\underline{U}_{L, t+1}, \underline{U}_{A, t + 1})$. 
\end{proof}

\subsection{Additional Results}
\begin{lemma}
    \label{lemma:induction-step}
    For any $t \in \{1, \tau - 1\}$, it follows that
    \begin{align}
        m_t(a_t, h_t) = \E\left[ \frac{g^d_{t + 1,\mathcal{B}}(A_{t+1} \mid H_{t+1})}{g_{t+1}(A_{t+1} \mid H_{t+1})} m_{t+1}(A_{t+1}, H_{t+1}) \mid A_t = a_t, H_t = h_t \right]. 
    \end{align}
\end{lemma}
\begin{proof}
    Begin by writing
    \begin{align}
        m_t(a_t, h_t) =& \E\left[ m_{t+1}(A_{t+1}^d \mid A_t = a_t, H_t = h_t, \underline{A}_{t+1} \in \underline{\mathcal{B}}_{t+1} \right] \qquad \text{(by def'n \eqref{eq:m_t})} \\
        =& \int_{a_{t+1}, l_{t+1}} m_{t+1}\left( d(a_{t+1}, h_{t+1}), h_{t+1} \right) d\P\left(a_{t+1},l_{t+1} | a_t, h_t, \underline{A}_{t+1} \in \underline{\mathcal{B}}_{t+1}\right) \\
        =& \int_{a_{t+1}, l_{t+1}} m_{t+1}\left( d(a_{t+1}, h_{t+1}), h_{t+1} \right) \frac{\P(\underline{A}_{t+1} \in \underline{\mathcal{B}}_{t+1} \mid a_{t+1}, h_{t+1})}{\P(\underline{A}_{t+1} \in \underline{\mathcal{B}}_{t+1} \mid a_{t}, h_{t})} d\P\left(a_{t+1},l_{t+1} | a_t, h_t \right) \qquad \text{(Bayes)} \\
        =& \int_{a_{t+1}, l_{t+1}} m_{t+1}\left( d(a_{t+1}, h_{t+1}), h_{t+1} \right) \\
        &\quad\quad \times\frac{\P(\underline{A}_{t+2} \in \underline{\mathcal{B}}_{t+2} \mid a_{t+1}, h_{t+1})\1(a_{t+1} \in \mathcal{B}_{t+1}) }{\P(\underline{A}_{t+1} \in \underline{\mathcal{B}}_{t+1} \mid a_{t}, h_{t})} d\P\left(a_{t+1},l_{t+1} | a_t, h_t \right). \label{eq:riesz-starting-point}
    \end{align}
    Next, consider separately the cases where $A_t$ is discrete and continuous.
    For discrete $A_t$, continue with
    \begin{align}
        \eqref{eq:riesz-starting-point} =& \int_{l_{t+1}} \sum_{a_{t+1}} m_{t+1}(d(a_{t+1}, h_{t+1}), h_{t+1}) \\
        &\qquad \frac{P(\underline{A}_{t+2} \in \underline{\mathcal{B}}_{t+2} \mid a_{t+1}, h_{t+1}) \1(a_{t+1}\in\mathcal{B}_{t+1})}{P(\underline{A}_{t+1} \in \underline{\mathcal{B}}_{t+1} \mid a_t, h_t)}
        g_{t+1}(a_{t+1} \mid h_{t+1}) d\P(l_{t+1} \mid a_{t}, l_{t}) \\
        =& \int_{l_{t+1}}  \sum_{a'_{t+1}} \sum_{a_{t+1}} \1\left\{ a_{t+1} \in \mathcal{B}_{t+1}, a'_{t+1} = d(a_{t+1}, h_{t+1}) \right\} m_{t+1}(a'_{t+1}), h_{t+1}) \\ 
        & \qquad \frac{P(\underline{A}_{t+2} \in \underline{\mathcal{B}}_{t+2} \mid a_{t+1}, h_{t+1}) \1(a_{t+1}\in\mathcal{B}_{t+1})}{P(\underline{A}_{t+1} \in \underline{\mathcal{B}}_{t+1} \mid a_t, h_t)} g_{t+1}(a_{t+1} \mid h_{t+1}) d\P(l_{t+1} \mid a_{t}, l_{t}) \\
        =& \int_{l_{t+1}} \sum_{a'_{t+1}} g_{t+1}(a'_{t+1} \mid h_{t+1}) m_{t+1}(a'_{t+1}, h_{t+1}) \frac{1}{g_{t+1}(a'_{t+1}\mid h_{t+1})} \sum_{a_{t+1}} \1\left\{ a_{t+1} \in \mathcal{B}_{t+1}, a'_{t+1} = d(a_{t+1}, h_{t+1}) \right\} \\
        & \qquad \frac{P(\underline{A}_{t+2} \in \underline{\mathcal{B}}_{t+2} \mid a_{t+1}, h_{t+1}) \1(a_{t+1}\in\mathcal{B}_{t+1})}{P(\underline{A}_{t+1} \in \underline{\mathcal{B}}_{t+1} \mid a_t, h_t)} g_{t+1}(a_{t+1} \mid h_{t+1}) d\P(l_{t+1} \mid a_{t}, l_{t}) \\
        &= \E\left[ \frac{g_{t+1,\mathcal{B}}^d(A_{t+1} \mid H_{t+1})}{g_{t+1}(A_{t+1} \mid H_{t+1})} m_{t+1}(A_{t+1}, H_{t+1}) \mid A_t = a_t, H_t = h_t \right].
    \end{align}
    For continuous $A_t$, similar logic applies, but we make use of Assumption~\ref{assumption:piecewise-invertibility}. Starting from \eqref{eq:riesz-starting-point}:
    {\small
    \begin{align}
        \eqref{eq:riesz-starting-point} =& \int_{l_{t+1}} \sum_{j=1}^{J_{t+1}(h_{t+1})} \int_{a_{t+1} \in \mathcal{I}_{j,t+1}} m_{t+1}\left( d(a_{t+1}, h_{t+1}), h_{t+1} \right) \\
        &\quad\quad \times \frac{\P(\underline{A}_{t+2} \in \underline{\mathcal{B}}_{t+2} \mid a_{t+1}, h_{t+1})\1(a_{t+1} \in \mathcal{B}_{t+1}) }{\P(\underline{A}_{t+1} \in \underline{\mathcal{B}}_{t+1} \mid a_{t}, h_{t})} d\P\left(a_{t+1},l_{t+1} | a_t, h_t \right) \\
        =& \int_{l_{t+1}} \sum_{j=1}^{J_{t+1}(h_{t+1})} \int_{a_{t+1} \in \mathcal{I}_{j,t+1}} m_{t+1}\left( d(a_{t+1}, h_{t+1}), h_{t+1} \right) \\
        & \quad\quad \times \frac{\P(\underline{A}_{t+2} \in \underline{\mathcal{B}}_{t+2} \mid a_{t+1}, h_{t+1})\1(a_{t+1} \in \mathcal{B}_{t+1}) }{\P(\underline{A}_{t+1} \in \underline{\mathcal{B}}_{t+1} \mid a_{t}, h_{t})} g_{t+1}(a_{t+1} \mid h_{t+1}) d\nu(a_2) d\P(l_2 \mid h_1) d\P\left(a_{t+1},l_{t+1} | a_t, h_t \right) \\
        =& \int_{l_{t+1}} \sum_{j=1}^{J_{t+1}(h_{t+1})} \int_{a'_{t+1} \in \mathcal{I}_{j,t+1}} m_{t+1}\left( a'_{t+1}, h_{t+1} \right) \\
        & \quad\quad \times \frac{\P(\underline{A}_{t+2} \in \underline{\mathcal{B}}_{t+2} \mid b_{j,t+1}(a'_{t+1}, h_{t+1}), h_{t+1})\1(b_{j,t+1}(a'_{t+1}, h_{t+1}) \in \mathcal{I}_{j, t+1}, b_{j,t+1}(a'_{t+1}, h_{t+1}) \in \mathcal{B}_{t+1}) }{\P(\underline{A}_{t+1} \in \underline{\mathcal{B}}_{t+1} \mid a_{t}, h_{t})} \\ 
        & \quad\quad \times g_{t+1}(b_{j,t+1}(a'_{t+1},h_{t+1}) \mid l_{t+1}, h_t) \left| b'_{j,2}(a'_2, h_2) \right| d\nu(a'_{t+1}) d\P(l_{t+1} \mid h_t) \\
        =& \E\left[ \frac{g^d_{t+1}(A_{t+1} \mid H_{t+1})}{g_{t+1}(A_{t+1} \mid H_{t+1})} m_{t+1}(A_{t+1}, H_{t+1}) \middle| A_t = a_t, H_t = h_t \right]. 
    \end{align}}
\end{proof}

\begin{lemma}
    \label{lemma:base-case}
    It holds that
    \begin{align}
        \E\left[ m_1(A_1^d, H_1) | \underline{A}_1 \in \underline{\mathcal{B}}_1 \right] = \E[\alpha_2(A_2, H_2) m_2(A_2, H_2)]. 
    \end{align}
\end{lemma}
\begin{proof}
    The proof begins similarly to the proof of Lemma~\ref{lemma:induction-step}, by writing
    \begin{align}
        \E[&m_1(A_1^d, H_1) \mid \underline{A}_1 \in \mathcal{B}_1] \\
        &= \int_{a_1, l_1} m_1(d(a_1, h_1), h_1) d\P(a_1, l_1 \mid \underline{A}_1 \in \underline{\mathcal{B}}) \\
        &= \int_{a_1, l_1} m_1(d(a_1, h_1), h_1) \frac{\P(\underline{A}_1 \in \underline{\mathcal{B}}_1 \mid a_1, l_1)}{\P(\underline{A}_1 \in \underline{\mathcal{B}}_1)}  d\P(a_1, l_1) \qquad \text{(Bayes)} \\
        &= \int_{a_1, l_1} m_1(d(a_1, h_1), h_1) \frac{\1(a_1 \in \mathcal{B}_1) \P(\underline{A}_2 \in \underline{\mathcal{B}}_2 \mid a_1, l_1)}{\P(\underline{A}_1 \in \underline{\mathcal{B}}_1)}  d\P(a_1, l_1) \label{eq:base-case-starting-point}.
    \end{align}

    Next, consider the case of discrete $A_1$:
    \begin{align}
        \eqref{eq:base-case-starting-point} &= \int_{l_1} \sum_{a_1} m_1(d(a_1, h_1), h_1) \frac{\1(a_1 \in \mathcal{B}_1) \P(\underline{A}_2 \in \underline{\mathcal{B}}_2 \mid a_1, l_1)}{\P(\underline{A}_1 \in \underline{\mathcal{B}}_1)}  d\P(a_1, l_1) \\
        &= \int_{l_1} \sum_{a'_1} \sum_{a_1} \1\{ a_1 \in \mathcal{B}_1, a'_1 = d(a_1, h_1) \} m_1(a'_1, h_1) \frac{\1(a_1 \in \mathcal{B}_1) P(\underline{A}_2 \in \underline{\mathcal{B}}_2 \mid a_1, l_1)}{P(\underline{A}_1 \in \underline{\mathcal{B}}_1)}  d\P(a_1, l_1) \\
        &= \int_{l_1} \sum_{a'_1} \sum_{a_1} \1\{ a_1 \in \mathcal{B}_1, a'_1 = d(a_1, h_1) \} m_1(a'_1, h_1) \frac{\1(a_1 \in \mathcal{B}_1) P(\underline{A}_2 \in \underline{\mathcal{B}}_2 \mid a_1, l_1)}{P(\underline{A}_1 \in \underline{\mathcal{B}}_1)} g_1(a_1 \mid l_1) d\P(l_1) \\
        &= \E\left[ \frac{g^d_{1, \mathcal{B}}(A_1 \mid H_1)}{g_1(A_1 \mid H_1)} m_1(A_1, H_1) \right] \\
        &= \E\left\{ \frac{g^d_{1,\mathcal{B}}(A_1 \mid H_1)}{g_1(A_1 \mid H_1)} \E\left[ \frac{g^d_{2, \mathcal{B}}(A_2 \mid H_2)}{g_2(A_2 \mid H_2)} m_2(A_2, H_2) \mid A_1, H_1 \right] \right\}\qquad \text{(by Lemma~\ref{lemma:induction-step})} \\
        &= \E\left[\alpha_2(A_2, H_2) m_2(A_2, H_2) \right]. 
    \end{align}
    For continuous $A_1$, continuing from \eqref{eq:base-case-starting-point}:
    \begin{align}
        \eqref{eq:base-case-starting-point} &= \int_{l_1} \sum_{j=1}^{J_1(h_1)} \int_{a_1 \in \mathcal{I}_{j,1}} m_1(d(a_1, h_1), h_1) \frac{\1(a_1 \in \mathcal{B}_1) \P(\underline{A}_2 \in \underline{\mathcal{B}}_2 | a_1, l_1)}{\P(\underline{A}_1 \in \underline{\mathcal{B}}_1)} d\P(a_1, l_1) \\
        &= \int_{l_1} \sum_{j=1}^{J_1(h_1)} \int_{a_1 \in \mathcal{I}_{j,1}} m_1(d(a_1, h_1), h_1) \frac{\1(a_1 \in \mathcal{B}_1) \P(\underline{A}_2 \in \underline{\mathcal{B}}_2 | a_1, l_1)}{\P(\underline{A}_1 \in \underline{\mathcal{B}}_1)} g_1(a_1 \mid h_1) d\nu(a_1) d\P(l_1) \\
        &= \int_{l_1} \sum_{j=1}^{J_1(h_1)} \int_{a_1'} m_1(a_1', h_1) \frac{\1(b_{j,1}(a'_1, h_1) \in \mathcal{I}_{t,1}(h_1), b_{j, 1}(a'_1, h_1) \in \mathcal{B}_1)\P(\underline{A}_2 \in \underline{\mathcal{B}}_2 \mid b_{j,1}(a'_1, h_1), l_1)}{\P(\underline{A}_1 \in \underline{\mathcal{B}}_1)} \\
        & \qquad \times g_1(b_{j,1}(a'_1, h_1) \mid l_1) \left| b'_{j,1}(a'_1, h_1) \right| d\nu(a_1) d\P(l_1) \\
        &= \E\left[ \frac{g^d_{1, \mathcal{B}}(A_1 \mid H_1}{g_1(A_1 \mid H_1)} m_1(A_1, H_1) \right] \\
        &= \E\left\{ \frac{g^d_{1,\mathcal{B}}(A_1 \mid H_1)}{g_1(A_1 \mid H_1)} \E\left[ \frac{g^d_{2, \mathcal{B}}(A_2 \mid H_2)}{g_2(A_2 \mid H_2)} m_2(A_2, H_2) \mid A_1, H_1 \right] \right\}\qquad \text{(by Lemma~\ref{lemma:induction-step})} \\
        &= \E\left[\alpha_2(A_2, H_2) m_2(A_2, H_2) \right]. 
    \end{align}
\end{proof}

\subsection{Proof of Lemma~\ref{lemma:riesz}}
\begin{proof}
    By Lemma~\ref{lemma:base-case}, we have that
    \begin{align}
        \E\left[ m_1(A_1^d, H_1) | \underline{A}_1 \in \underline{\mathcal{B}}_1 \right] = \E[\alpha_2(A_2, H_2) m_2(A_2, H_2)]. 
    \end{align}
    Repeated application of Lemma~\ref{lemma:induction-step} yields the result, using that at time $t = \tau$,
    \begin{align}
        \E[\alpha_\tau(A_\tau, H_\tau) m_\tau(A_\tau, H_\tau)] = \E[\alpha_\tau(A_\tau, H_\tau) Y]. 
    \end{align}
\end{proof}

\subsection{von-Mises expansion}
\label{appendix:von-mises}
As introduced in the main text, the von-Mises expansion of $\theta$ is given by, for any $\P, \F$ in the model $\mathcal{M}$,
\begin{align}
    \theta(\P)  = \theta(\F) - \E_{\F}\{\D(Z; \P)\} + \Rem(\P, \F).
\end{align}
In this appendix we derive the form of the EIF $\D$ and the second-order term $\Rem$. 

Begin by defining, for any two distributions $\P \in \mathcal{M}$ and $\F \in \mathcal{M}$, the functional
\begin{align}
    m_{t+1,\P,\F}(a_t, h_t) = \E_\P\{b_{t+1}(A_{t+1}, H_{t+1};m_{t+1,\F})\mid A_t=a_t, H_t=h_t,\underline{A}_{t+1} \in \underline{\mathcal{B}}_{t+1}\}.
\end{align}
Then, according to the definition
{\small
\begin{align}
    \Psi_t(m_{t+1})=\E_\P\bigg\{\E_\P\big\{\E_\P[b_{t+1}(A_{t+1}, H_{t+1}; m_{t+1})\mid A_{t}=A_{t}^d, H_{t}, \underline{A}_{t+1} \in \underline{\mathcal{B}}_{t+1}]\mid A_{t-1} = A_{t-1}^d,  H_{t-1}, \underline{A}_{t} \in \underline{\mathcal{B}}_{t}\big\}\cdots\bigg\},
\end{align}
}
we have $\Psi_t(m_{t+1,\F})=\Psi_{t-1}(m_{t, \P,\F})$. Telescoping yields
\[\Psi_0(m_{1,\F}) - \Psi_{\tau}(m_{\tau+1,\F})=-\sum_{t=0}^{\tau}\{\Psi_t(m_{t+1,\P,\F}) - \Psi_t(m_{t+1,\F})\}.\]
Notice that
{\small
\begin{align}
    \label{eq:trans}
    \E_\P[\alpha_t(A_t, H_t)m_{t,\P,\F}(A_t, H_t)] & = \E_\P\{\alpha_t(A_t, H_t)\E_\P[m_{t+1, \F}(A_{t+1}^d, H_{t+1})\mid A_t, H_t, \underline{A}_{t+1} \in \underline{\mathcal{B}}_{t+1}]\}\\
    &=\E_\P\left[\frac{\one\{\underline{A}_{t+1} \in \underline{\mathcal{B}}_{t+1}\} }{G_{t,\P}(A_t, H_t) }\alpha_t(A_t, H_t)b_{t+1}(A_{t+1}, H_{t+1}; m_{t+1,\F})\right],    
\end{align}}%
as well as
{\small
\begin{align}
  \E_\P[\alpha_t(A_t, H_t)m_{t,\F}(A_t, H_t)]=\E_\P\left[\frac{\one\{\underline{A}_{t+1} \in \underline{\mathcal{B}}_{t+1}\} }{G_{t,\P}(A_t, H_t) }\alpha_t(A_t, H_t)m_{t,\F}(A_t, H_t)\right].   
\end{align}}%
Using these two and Lemma~\ref{lemma:riesz} above we have
{\small
\begin{align*}
    \Psi_{t}&(m_{t,\P,\F}) - \Psi_{t}(m_{t,\F})  = \E_\P[\alpha_{t,\P}(A_t, H_t)\{m_{t,\P,\F}(A_t, H_t) - m_{t,\F}(A_t, H_t)\}]\\
    & = \E_\P[\{\alpha_{t,\P}(A_t, H_t)-\alpha_{t,\F}(A_t, H_t)\}\{m_{t,\P,\F}(A_t, H_t) - m_{t,\F}(A_t, H_t)\}]\\
        &+\E_\P\left[\alpha_{t,\F}(A_t, H_t)\left\{\frac{\one\{\underline{A}_{t+1} \in \underline{\mathcal{B}}_{t+1}\} }{G_{t,\P}(A_t, H_t) } - \frac{\one\{\underline{A}_{t+1} \in \underline{\mathcal{B}}_{t+1}\} }{G_{t,\F}(A_t, H_t) }\right\}\{b_{t+1}(A_{t+1}, H_{t+1};m_{t+1, \F}) - m_{t,\F}(A_t, H_t)\}\right]\\
    &+\E_\P\left[\alpha_{t,\F}(A_t, H_t)\frac{\one\{\underline{A}_{t+1} \in \underline{\mathcal{B}}_{t+1}\} }{G_{t,\F}(A_t, H_t) }\{b_{t+1}(A_{t+1}, H_{t+1};m_{t+1, \F}) - m_{t,\F}(A_t, H_t)\}\right].
\end{align*}}%
Define the second-order term
\begin{align*}
    R(\P,\F) =& -\sum_{t=1}^\tau\E_\P[\{\alpha_{t,\P}(A_t, H_t)-\alpha_{t,\F}(A_t, H_t)\}\{m_{t,\P,\F}(A_t, H_t) - m_{t,\F}(A_t, H_t)\}]\\
    &-\sum_{t=1}^\tau\E_\P\left[\alpha_{t,\F}(A_t, H_t)\left\{1 - \frac{G_{t,\P}(A_t, H_t) }{G_{t,\F}(A_t, H_t) }\right\}\{m_{t,\P,\F}(A_t, H_t) - m_{t,\F}(A_t, H_t)\}\right],
\end{align*}
where we note that the terms corresponding to $t=0$ and $t=\tau+1$ vanish since $\alpha_{0,\F}=1$ and $m_{\tau+1,\F}=Y$ for all $\F$ by definition. Notice that $\Psi_0(m_{0,\F})=\theta(\F)$ and $\Psi_{\tau+1}(m_{\tau+1,\F})=\theta(\P)$. Putting the above together yields
\[\theta(\F) - \theta(\P)=-\E_\P[\D(X;\eta_\F)] + R(\P,\F),\]
where
\[\D(X;\eta_\F) = \sum_{t=0}^{\tau}\alpha_{t,\F}(A_t, H_t)\frac{\one\{\underline{A}_{t+1} \in \underline{\mathcal{B}}_{t+1}\} }{G_{t,\P}(A_t, H_t) }\{b_{t+1}(A_{t+1}, H_{t+1};m_{t+1, \F}) - m_{t,\F}(A_t, H_t)\}\]
and $\eta_\F=(\alpha_{t,\F},m_{t,\F},G_{t,\F}:t=0,\ldots, \tau)$.

\subsection{Riesz empirical loss minimization}
\label{appendix:riesz-empirical-loss}
Notice that 
{\small
\begin{align*}
\Psi_t(\tilde\alpha)&=\E\bigg\{\E\big\{\E[b_{t+1}(A_{t+1}, H_{t+1}; \tilde\alpha)\mid A_{t}=A_{t}^d, H_{t}, \underline{A}_{t+1} \in \underline{\mathcal{B}}_{t+1}]\mid A_{t-1} = A_{t-1}^d,  H_{t-1}, \underline{A}_{t} \in \underline{\mathcal{B}}_{t}\big\}\cdots\bigg\}    \\
&=\Psi_{t-1}(q_{t}),
\end{align*}}%
where
\[q_{t}(A_{t}, H_{t}) = \E[b_{t+1}(A_{t+1}, H_{t+1}; \tilde\alpha)\mid A_{t}, H_{t}, \underline{A}_{t+1} \in \underline{\mathcal{B}}_{t+1}].\]
Thus, in light of (\ref{eq:trans})
\begin{align*}
    \Psi_t(\tilde\alpha) &= \E[\alpha_{t-1}(A_{t-1}, H_{t-1}) q_{t-1}(A_{t-1}, H_{t-1})]\\
    &= \E\left[\alpha_{t-1}(A_{t-1}, H_{t-1})\frac{\one\{\underline{A}_{t} \in \underline{\mathcal{B}}_{t}\} }{G_{t-1,\P}(A_{t-1}, H_{t - 1}) } b_t(A_t, H_t; \tilde\alpha)\right].
\end{align*}
Thus, we have

\begin{align}
    \alpha_t =\argmin_{\tilde \alpha}\E\left\{\tilde\alpha(A_t,H_t) - \alpha_{t-1}(A_{t-1}, H_{t-1}) \frac{\one\{\underline{A}_{t} \in \underline{\mathcal{B}}_{t}\} }{G_{t-1,\P}(A_{t-1}, H_{t-1}) }b_t(A_t, H_t; \tilde\alpha)\right\}.
\end{align}

\subsection{Proof of Theorem~\ref{theorem:tmle-weak-convergence}}
Theorem~\ref{theorem:tmle-weak-convergence} can be proven following the same steps as \citealt[Theorem 3]{diaz2023lmtp}. 

\subsection{Riesz representation estimation details}
\label{appendix:nn-details}
We used the \texttt{SuperRiesz} \texttt{R} package \citep{susmann2024superriesz} to implement the Riesz estimation strategy described in Section~\ref{section:estimation}.
For Simulation Study 1 we used the \texttt{glm} learner with all first order interaction terms included. For Simulation Study 2 we used the \texttt{nn} (neural network) learner with 500 training epochs, a learning rate parameter set to $1 \times 10^{-3}$, and 25 neurons in 1 hidden layer.

\subsection{Additional Simulation Results}
\label{appendix:additional-results}
    \begin{longtable}{|rrrrrrrrrr|}
        \hline
        & & & 95\% Coverage & \multicolumn{3}{c}{MAE $\times$ 100} & \multicolumn{3}{c|}{ME $\times$ 100} \\
        $\mathcal{B}^a$ & Scenario & $N$ & TMLE & IPW &  Sub & TMLE & IPW & Sub & TMLE  \\
        \hline
        \{ 0 \} & 1 & 500 & 97.0\% & 10.35 & 4.51 & 4.78 & -9.14 & -2.00 & 0.04\\
         &  & 1000 & 97.5\% & 5.02 & 2.52 & 2.99 & -3.97 & -0.71 & -0.60\\
         &  & 2500 & 94.5\% & 2.79 & 1.75 & 2.24 & -1.65 & -0.25 & -0.30\\
         &  & 5000 & 97.5\% & 1.61 & 1.13 & 1.36 & -0.98 & -0.12 & -0.14\\
         & 2 & 500 & 68.0\% & 22.15 & 4.41 & 4.74 & -22.15 & -2.33 & -2.34\\
         &  & 1000 & 78.5\% & 22.42 & 2.46 & 2.70 & -22.42 & -0.72 & -0.86\\
         &  & 2500 & 73.0\% & 22.19 & 1.71 & 1.80 & -22.19 & -0.13 & -0.21\\
         &  & 5000 & 76.0\% & 22.24 & 1.09 & 1.16 & -22.24 & -0.08 & -0.12\\
         & 3 & 500 & 68.5\% & 13.44 & 22.20 & 7.60 & 2.30 & -22.20 & -3.38\\
         &  & 1000 & 63.0\% & 13.12 & 22.43 & 7.34 & 2.90 & -22.43 & -3.69\\
         &  & 2500 & 53.5\% & 11.55 & 22.20 & 5.68 & 4.88 & -22.20 & -3.36\\
         &  & 5000 & 40.5\% & 8.57 & 22.25 & 5.47 & 3.73 & -22.25 & -4.58\\
         & 4 & 500 & 0.0\% & 22.15 & 22.19 & 22.19 & -22.15 & -22.19 & -22.19\\
         &  & 1000 & 0.0\% & 22.42 & 22.44 & 22.44 & -22.42 & -22.44 & -22.44\\
         &  & 2500 & 0.0\% & 22.19 & 22.20 & 22.20 & -22.19 & -22.20 & -22.20\\
         &  & 5000 & 0.0\% & 22.24 & 22.25 & 22.25 & -22.24 & -22.25 & -22.25\\
        \{ 1 \} & 1 & 500 & 97.5\% & 18.66 & 4.53 & 6.39 & -17.59 & -0.96 & 0.03\\
         &  & 1000 & 97.5\% & 6.32 & 3.07 & 4.32 & -4.23 & -0.01 & 0.66\\
         &  & 2500 & 97.0\% & 3.16 & 1.82 & 2.47 & -2.04 & -0.01 & -0.05\\
         &  & 5000 & 94.0\% & 2.04 & 1.41 & 1.92 & -1.12 & -0.06 & -0.11\\
         & 2 & 500 & 74.5\% & 15.68 & 4.41 & 4.60 & -15.68 & -1.43 & -1.19\\
         &  & 1000 & 76.0\% & 15.95 & 2.91 & 3.23 & -15.95 & -0.09 & -0.23\\
         &  & 2500 & 77.0\% & 15.72 & 1.79 & 1.93 & -15.72 & 0.12 & 0.01\\
         &  & 5000 & 74.5\% & 15.77 & 1.35 & 1.52 & -15.77 & 0.04 & -0.04\\
         & 3 & 500 & 74.5\% & 24.36 & 15.72 & 7.60 & -9.63 & -15.72 & -2.94\\
         &  & 1000 & 76.5\% & 14.42 & 15.97 & 7.34 & -1.74 & -15.97 & -4.42\\
         &  & 2500 & 62.0\% & 14.76 & 15.73 & 6.60 & -2.87 & -15.73 & -4.75\\
         &  & 5000 & 53.5\% & 11.85 & 15.77 & 5.63 & -1.89 & -15.77 & -4.79\\
         & 4 & 500 & 0.0\% & 15.68 & 15.73 & 15.72 & -15.68 & -15.73 & -15.72\\
         &  & 1000 & 0.0\% & 15.95 & 15.96 & 15.96 & -15.95 & -15.96 & -15.96\\
         &  & 2500 & 0.0\% & 15.72 & 15.73 & 15.73 & -15.72 & -15.73 & -15.73\\
         &  & 5000 & 0.0\% & 15.77 & 15.77 & 15.77 & -15.77 & -15.77 & -15.77\\
        \{ 2 \} & 1 & 500 & 100.0\% & 138.62 & 5.15 & 6.35 & -138.62 & -3.57 & -1.61\\
         &  & 1000 & 100.0\% & 44.16 & 3.08 & 4.63 & -44.16 & -0.82 & -0.19\\
         &  & 2500 & 96.0\% & 7.83 & 1.91 & 3.43 & -7.36 & 0.00 & 0.03\\
         &  & 5000 & 97.5\% & 2.77 & 1.37 & 2.37 & -1.24 & -0.18 & -0.31\\
         & 2 & 500 & 75.5\% & 18.26 & 5.65 & 5.37 & -18.26 & -4.55 & -3.64\\
         &  & 1000 & 84.0\% & 18.52 & 3.18 & 3.37 & -18.52 & -1.23 & -0.90\\
         &  & 2500 & 85.0\% & 18.30 & 1.80 & 2.12 & -18.30 & 0.03 & -0.13\\
         &  & 5000 & 80.5\% & 18.35 & 1.35 & 1.47 & -18.35 & -0.15 & -0.12\\
         & 3 & 500 & 97.0\% & 65.12 & 18.28 & 12.87 & -64.95 & -18.28 & -11.30\\
         &  & 1000 & 89.0\% & 43.51 & 18.53 & 12.02 & -42.94 & -18.53 & -11.53\\
         &  & 2500 & 34.5\% & 32.11 & 18.30 & 12.24 & -31.91 & -18.30 & -12.12\\
         &  & 5000 & 3.0\% & 26.31 & 18.35 & 12.41 & -26.31 & -18.35 & -12.41\\
         & 4 & 500 & 0.0\% & 18.26 & 18.27 & 18.28 & -18.26 & -18.27 & -18.28\\
         &  & 1000 & 0.0\% & 18.52 & 18.53 & 18.53 & -18.52 & -18.53 & -18.53\\
         &  & 2500 & 0.0\% & 18.30 & 18.30 & 18.30 & -18.30 & -18.30 & -18.30\\
         &  & 5000 & 0.0\% & 18.35 & 18.35 & 18.35 & -18.35 & -18.35 & -18.35\\
        $\{ 3 \}$ & 1 & 500 & 100.0\% & 203.60 & 4.27 & 6.65 & -203.51 & -2.13 & -0.71\\
         &  & 1000 & 100.0\% & 71.01 & 3.99 & 6.18 & -70.67 & -0.90 & -1.04\\
         &  & 2500 & 97.5\% & 17.98 & 2.32 & 4.09 & -17.69 & 0.40 & 0.41\\
         &  & 5000 & 94.5\% & 5.27 & 1.68 & 3.00 & -3.77 & 0.28 & 0.00\\
         & 2 & 500 & 88.0\% & 10.11 & 4.70 & 4.33 & -10.11 & -3.49 & -1.82\\
         &  & 1000 & 76.5\% & 10.37 & 3.84 & 4.50 & -10.37 & -1.32 & -0.81\\
         &  & 2500 & 72.0\% & 10.15 & 2.12 & 2.50 & -10.15 & 0.42 & 0.42\\
         &  & 5000 & 70.0\% & 10.20 & 1.57 & 1.80 & -10.20 & 0.19 & 0.31\\
         & 3 & 500 & 98.0\% & 62.81 & 10.10 & 10.90 & -62.81 & -10.10 & -9.30\\
         &  & 1000 & 90.5\% & 50.56 & 10.37 & 10.97 & -50.49 & -10.37 & -10.33\\
         &  & 2500 & 77.0\% & 26.11 & 10.15 & 8.50 & -26.11 & -10.15 & -8.40\\
         &  & 5000 & 48.5\% & 19.13 & 10.20 & 7.04 & -19.13 & -10.20 & -7.03\\
         & 4 & 500 & 0.5\% & 10.11 & 10.13 & 10.10 & -10.11 & -10.13 & -10.10\\
         &  & 1000 & 0.0\% & 10.37 & 10.38 & 10.37 & -10.37 & -10.38 & -10.37\\
         &  & 2500 & 0.0\% & 10.15 & 10.15 & 10.15 & -10.15 & -10.15 & -10.15\\
         &  & 5000 & 0.0\% & 10.20 & 10.20 & 10.20 & -10.20 & -10.20 & -10.20\\
        $\{ 4 \}$ & 1 & 500 & 100.0\% & 288.70 & 4.39 & 6.83 & -287.92 & -1.30 & -0.95\\
         &  & 1000 & 99.0\% & 98.38 & 3.95 & 6.16 & -97.77 & -0.61 & -0.43\\
         &  & 2500 & 95.0\% & 20.27 & 3.03 & 4.70 & -17.83 & -0.53 & 0.06\\
         &  & 5000 & 95.0\% & 5.60 & 1.84 & 2.92 & -0.36 & -0.29 & -0.53\\
         & 2 & 500 & 83.0\% & 2.20 & 4.69 & 4.61 & -1.50 & -2.86 & -0.70\\
         &  & 1000 & 68.0\% & 2.03 & 3.62 & 4.49 & -1.76 & -1.26 & -0.42\\
         &  & 2500 & 61.5\% & 1.57 & 2.74 & 3.19 & -1.53 & -0.54 & -0.52\\
         &  & 5000 & 60.5\% & 1.58 & 1.54 & 1.89 & -1.58 & -0.40 & -0.29\\
         & 3 & 500 & 97.5\% & 46.00 & 2.19 & 8.92 & -45.59 & -1.49 & -4.11\\
         &  & 1000 & 94.5\% & 18.07 & 2.03 & 7.30 & -16.47 & -1.75 & -3.15\\
         &  & 2500 & 87.0\% & 7.18 & 1.56 & 5.26 & 0.76 & -1.53 & -2.56\\
         &  & 5000 & 86.0\% & 8.11 & 1.58 & 3.28 & 7.50 & -1.58 & -0.48\\
         & 4 & 500 & 92.5\% & 2.20 & 2.19 & 2.18 & -1.50 & -1.48 & -1.47\\
         &  & 1000 & 80.0\% & 2.03 & 2.02 & 2.02 & -1.76 & -1.74 & -1.74\\
         &  & 2500 & 69.5\% & 1.57 & 1.56 & 1.56 & -1.53 & -1.53 & -1.53\\
         &  & 5000 & 38.5\% & 1.58 & 1.58 & 1.58 & -1.58 & -1.58 & -1.58\\
        $\{ 5 \}$ & 1 & 500 & 99.0\% & 207.72 & 4.20 & 7.40 & -205.55 & 2.00 & 0.03\\
         &  & 1000 & 97.0\% & 76.31 & 3.62 & 6.49 & -71.48 & 0.33 & 0.32\\
         &  & 2500 & 93.5\% & 13.88 & 2.40 & 5.00 & -5.56 & 0.74 & 0.99\\
         &  & 5000 & 92.5\% & 8.28 & 1.77 & 2.75 & 7.39 & 0.55 & 0.83\\
         & 2 & 500 & 74.0\% & 10.37 & 4.25 & 4.92 & 10.37 & 0.99 & 3.11\\
         &  & 1000 & 72.5\% & 10.11 & 3.73 & 3.87 & 10.11 & 0.08 & 0.83\\
         &  & 2500 & 66.0\% & 10.34 & 2.25 & 2.41 & 10.34 & 0.61 & 0.74\\
         &  & 5000 & 59.5\% & 10.29 & 1.74 & 1.77 & 10.29 & 0.71 & 0.62\\
         & 3 & 500 & 97.5\% & 24.82 & 10.42 & 9.07 & -19.95 & 10.42 & 1.34\\
         &  & 1000 & 84.5\% & 14.76 & 10.14 & 7.40 & 5.07 & 10.14 & 0.70\\
         &  & 2500 & 71.0\% & 21.05 & 10.35 & 4.59 & 19.84 & 10.35 & 1.37\\
         &  & 5000 & 72.5\% & 21.70 & 10.29 & 2.90 & 21.70 & 10.29 & 1.73\\
         & 4 & 500 & 0.0\% & 10.37 & 10.41 & 10.42 & 10.37 & 10.41 & 10.42\\
         &  & 1000 & 0.0\% & 10.11 & 10.14 & 10.13 & 10.11 & 10.14 & 10.13\\
         &  & 2500 & 0.0\% & 10.34 & 10.35 & 10.35 & 10.34 & 10.35 & 10.35\\
         &  & 5000 & 0.0\% & 10.29 & 10.29 & 10.29 & 10.29 & 10.29 & 10.29\\
        \hline
    \caption{Results of Simulation Study 1 showing empirical coverage of the 95\% confidence intervals, Mean Absolute Error (MAE), and Mean Error (ME) for the inverse probability weighted estimator (IPW), substitution estimator (Sub), and Targeted minimum loss-based estimator (TMLE).}
    \label{tab:simulation-results-1-appendix}
\end{longtable}
\fi

\bibliography{refs}

\end{document}